\documentclass[aos, preprint]{imsart}
\usepackage{amsfonts}

\usepackage{amsmath}
\usepackage{enumerate}
\usepackage{mathabx}
\usepackage{listings}
\usepackage{graphicx}
\usepackage{color}

\usepackage{bm}
\usepackage{verbatim}

\newcommand{\handout}[5]{
   \renewcommand{\thepage}{#1-\arabic{page}}
   \noindent
   \begin{center}
   \framebox{
      \vbox{
    \hbox to 5.78in { {\bf 6.893 Sub-linear Algorithms}
     	 \hfill #2 }
       \vspace{4mm}
       \hbox to 5.78in { {\Large \hfill #5  \hfill} }
       \vspace{2mm}
       \hbox to 5.78in { {\it #3 \hfill #4} }
      }
   }
   \end{center}
   \vspace*{4mm}
}

\newtheorem{theorem}{Theorem}
\newtheorem{lemma}{Lemma}
\newtheorem{corollary}{Corollary}

\newtheorem{proposition}[theorem]{Proposition}

\newcommand{\qed}{\rule{7pt}{7pt}}

\newenvironment{proof}{\noindent{\bf Proof}\hspace*{1em}}{\qed\bigskip}
\newenvironment{proof-sketch}{\noindent{\bf Sketch of Proof}\hspace*{1em}}{\qed\bigskip}
\newenvironment{proof-idea}{\noindent{\bf Proof Idea}\hspace*{1em}}{\qed\bigskip}
\newenvironment{proof-of-lemma}[1]{\noindent{\bf Proof of Lemma #1}\hspace*{1em}}{\qed\bigskip}
\newenvironment{proof-attempt}{\noindent{\bf Proof Attempt}\hspace*{1em}}{\qed\bigskip}

\newenvironment{remark}{\noindent{\bf Remark}\hspace*{1em}}{\bigskip}

\newenvironment{proof-of-theorem}[1]{\noindent{\bf Proof of Theorem #1}\hspace*{1em}}{\qed\bigskip}

\newenvironment{proof-of-corollary}[1]{\noindent{\bf Proof of Corollary #1}\hspace*{1em}}{\qed\bigskip}

\newenvironment{remarks}{\noindent{\bf Remarks}\hspace*{1em}}{\bigskip}

\makeatletter
\def\fnum@figure{{\bf Figure \thefigure}}
\def\fnum@table{{\bf Table \thetable}}
\long\def\@mycaption#1[#2]#3{\addcontentsline{\csname
  ext@#1\endcsname}{#1}{\protect\numberline{\csname
  the#1\endcsname}{\ignorespaces #2}}\par
  \begingroup
    \@parboxrestore
    \small
    \@makecaption{\csname fnum@#1\endcsname}{\ignorespaces #3}\par
  \endgroup}
\def\mycaption{\refstepcounter\@captype \@dblarg{\@mycaption\@captype}}
\makeatother

\newcommand{\mathify}[1]{\ifmmode{#1}\else\mbox{$#1$}\fi}
\newcommand{\bigO}O

\usepackage{enumitem}
\usepackage{hyperref}

\hypersetup{colorlinks, 
			citecolor=blue, 
			linkcolor=blue, 
			urlcolor=black}

\newcommand{\bb}[1]{\bm{#1}}
\usepackage[toc,page]{appendix}

\RequirePackage[OT1]{fontenc}

\usepackage[space]{grffile}
\usepackage{subfig}

\setattribute{journal}{name}{}

\usepackage{natbib}
\usepackage{xr}

\begin{document}

\title{Unifying approach to selective inference with applications to cross-validation}

\runtitle{Unifying approach to selective inference with applications to CV}

\begin{aug}
\author{\fnms{Jelena} \snm{Markovic}\corref{}\thanksref{t1}\ead[label=e1]{jelenam@stanford.edu},  
\fnms{Lucy} \snm{Xia}\ead[label=e2]{lucyxia@stanford.edu},
\fnms{Jonathan} \snm{Taylor}\thanksref{t2}\ead[label=e3]{jonathan.taylor@stanford.edu}
}
\runauthor{Markovic et al.}
\affiliation{Stanford University}
\address{Department of Statistics\\ Stanford University 
\\ Sequoia Hall \\ Stanford, CA 94305, USA\\ \texttt{$\{$jelenam, lucyxia, jonathan.taylor$\}$@stanford.edu}}
\end{aug}

\thankstext{t1}{Supported by Stanford Graduate Fellowship.}
\thankstext{t2}{Supported in part by National Science Foundation grant DMS-1208857 and Air Force Office of Sponsored Research grant 113039.}

\begin{abstract} 

We develop tools to do valid post-selective inference for a family of model selection procedures, including choosing a model via cross-validated Lasso. The tools apply universally when the following random vectors are jointly asymptotically multivariate Gaussian: 1. the vector composed of each model's quality value evaluated under certain model selection criteria (e.g.~cross-validation errors across folds, AIC, prediction errors etc.) 2. the test statistics from which we make inference on the parameters; it is worth noting that the parameters here are chosen after model selection methods are performed. Under these assumptions, we derive a pivotal quantity that has an asymptotically Unif$(0,1)$ distribution which can be used to perform tests and construct confidence intervals. Both the tests and confidence intervals are selectively valid for the chosen parameter. While the above assumptions may not be satisfied in some applications, we propose a novel variation to these model selection procedures by adding Gaussian randomizations to either one of the two vectors. As a result, the joint distribution of the above random vectors is multivariate Gaussian and our general tools apply. We illustrate our method by applying it to four important procedures for which very few selective inference results have been developed: cross-validated Lasso, cross-validated randomized Lasso, AIC-based model selection among a fixed set of models and inference for a newly introduced novel marginal LOCO parameter, inspired by the LOCO parameter of \cite{rinaldo2016bootstrapping}; and we provide complete results for these cases.  
For randomized model selection procedures, we develop Markov chain Monte Carlo sampling scheme to construct valid post-selective confidence intervals empirically. 

\end{abstract} 

\maketitle

\section{Introduction}

\textit{Selective inference} has attracted a lot of research interest in recent years. Intuitively, if we find important variables/associations after performing statistical learning methods on a set of data, to conduct proper inference on the selected variables or assess the strength of the associations, we should adjust for the selective procedure, since we have ``searched for/cherry-picked'' these variables \citep{berk2013valid, taylor2015statistical}. For a few model selection procedures, valid post-selective inferences have been developed, but it is rare that any of these methods can be applied universally. Recall that general model selection procedures can be briefly described as the following: we start with a measure of model quality, such as the likelihood, AIC, BIC, or prediction errors. After calculating each model's quality value under a specific measure (or criterion), we select the best model that either attains the minimum or the maximum of these values. Once we pick a model and choose a corresponding parameter of interest, we build a test statistic conditional on the selected model and use it to do inference on the parameter. We see that to do valid selective inference, we must account for observing the minimizer of the vector of all models' quality values.

In this paper, we develop tools to do post-selective inference that apply for general rank-based model selection procedures. The technique applies universally as long as the joint distribution of the vector of the models' quality values and the test statistic is asymptotically a multivariate Gaussian distribution. However, as we will illustrate later in the paper, these multivariate Gaussian assumptions do not hold in general. As a remedy, we propose a novel way to Gaussian-ize either the quality values vector or the test statistic by adding small Gaussian randomizations to it. In this way, we will show that we are able to carry out valid inference without losing model selection properties by added randomization. We apply randomization and our technical tools to four widely-used cases. The first three of them, cross-validated Lasso, cross-validated randomized Lasso, and AIC-based model selection over a fixed set of models, are examples when the vector of models' quality values itself is not asymptotically jointly Gaussian. The fourth one, inference after marginal LOCO parameter, defined in this work, is an example when the test statistics does not follow Gaussian distribution. Let us start with a brief introduction of these selection procedures.

\begin{enumerate}
\item \textbf{Cross-validated Lasso:} Among many learning methods, variable selection via Lasso has been one of the most popular ones. It describes the following procedure: given the data $(X,y)\in\mathbb{R}^{n\times p}\times\mathbb{R}^n$, we choose a set of important variables as the non-zero set in $\hat{\beta}=\hat{\beta}(X,y,\lambda)$, where

\begin{equation} \label{eq:lasso:objective:simple}
	\hat{\beta}(X,y,\lambda) = \textnormal{arg}\underset{\beta\in\mathbb{R}^p}{\min}\:\frac{1}{2}\left\|y-X\beta\right\|_2^2+\lambda\|\beta\|_1. 
\end{equation}
The above objective induces sparsity in the solution $\hat{\beta}$, so we denote with $\widehat{E}=\widehat{E}(X,y,\lambda)$ the set of the non-zero coefficients of $\hat{\beta}(X,y,\lambda)$. The goal is to provide valid inference for some model parameters chosen after observing $\widehat{E}=E$, where $E$ represents the realized, or observed, set of the selected predictors. So far, all previous attempts to address this problem, including \citealt{lee2013exact, lee_screening, sequential_post_selection}, assume fixed (pre-specified) $\lambda$ and typically fixed $X$. Let us denote with $\hat{\beta}_E=\hat{\beta}(X,y,\lambda)_E$ and $\hat{\beta}_{-E}=\hat{\beta}(X,y,\lambda)_{-E}$ the active (non-zero) and inactive sub-vector of $\hat{\beta}$ respectively, and the signs of the active sub-vector with $s_E=\textnormal{sign}(\hat{\beta}_E)$. Treating $\lambda$ as a constant, the Lasso selection event can be represented as 
\begin{equation} \label{eq:lasso:selection:fixed:lambda}
\begin{aligned}
	(X,y)\in\mathcal{S}_{E,s_E}
	&=\left\{(X',y') \in\mathbb{R}^{n\times p}\times\mathbb{R}^n: \widehat{E}(X',y',\lambda)=E, \textnormal{sign}(\hat{\beta}(X',y',\lambda)_E) = s_E \right\} \\
	&=\left\{(X',y')\in\mathbb{R}^{n\times p}\times\mathbb{R}^n:\hat{\beta}(X',y',\lambda)_{-E}=0, \textnormal{sign}(\hat{\beta}(X',y',\lambda)_E)=s_E\right\},
\end{aligned}
\end{equation}
and we need to do inference conditional on this event. Note that all the previous works, including our current work, condition on both the set of selected predictors $E$ and the observed signs $s_E$.

Fixing $\lambda$ implies the selection event $\mathcal{S}_{E,s_E}$ is affine in the response vector $y$ \citep{lee2013exact}. However, in practice, the optimal $\lambda$ is not given automatically, but usually a result of cross-validation (CV). Therefore, it is necessary to adjust for cross-validation to perform valid selective inference.

\citealt{loftus_cv} tried to perform post-selection inference after adjusting for cross-validation by writing out explicitly the quadratic constraints coming from cross-validation. This approach conditions not only on the minimizer of the cross-validated error vector but also on the $K\times(\textnormal{the grid size for }\lambda)$ intermediate models produced through $K$-fold CV. This is redundant in practice, since an analyst usually does not look at these intermediate models. Extra conditioning leaves less information for inference, reducing the statistical power of the tests performed.

In this paper, we provide a cleaner solution by conditioning only on the final model chosen by cross-validated Lasso, in addition to the constraint on vector of cross-validated errors introduced by the CV procedure. The latter event accounts for the fact that the penalty level $\lambda$ is chosen by looking at the minimizer of the cross-validated errors across folds.
Specifically, we combine Lasso optimization with cross-validation as one selection event and construct a selective pivot, which is a test statistics valid post-selection inference. To apply our technique, for a given grid of $\lambda$ values, we need the corresponding cross-validated error vector to be jointly asymptotically normal with the data. This is not true with vanilla cross-validation. Therefore, we propose a randomized version of cross-validation to fulfill our goal.

\item \textbf{Cross-validated randomized Lasso:} As illustrated in \citealt{tian2016magic, selective_sampler}, adding randomization to the model selection procedure, such as Lasso, greatly enhances power. Meanwhile, with the help of added randomization in the Lasso objective, the selection region simplifies compared to the one in \citealt{lee2013exact} and hence enables us to easily adapt Markov chain Monte Carlo (MCMC) techniques when sampling from post-selection density. Due to these benefits, we report post-selective $p$-values and construct confidence intervals after cross-validated randomized Lasso. We borrow techniques from \citealt{bootstrap_mv}, where they perform valid inference after running several model selection procedures on the same data set, with each procedure called a view/query. 

\item \textbf{AIC-based model selection among a given set of models:} In this application, the vector of models' quality values are their prediction errors and similar to the above two examples, this vector is asymptotically generally not multivariate Gaussian. 

\end{enumerate}

A crucial ingredient in the examples above is adding randomization to the curve composed of models' quality values, whether it is cross-validation curve or AIC criteria curve, to make this vector asymptotically jointly Gaussian. The test statistic used also needs to be asymptotically Gaussian pre-selection, hence we need to make sure this requirement is also satisfied. The usual selective inference parameters are population regression parameters corresponding to the selected model. In this case the test statistic used is the least squares estimator, hence asymptotically Gaussian pre-selection (treating the selected model as fixed in advance and non-random) under mild conditions. However, for other parameters the choice for the corresponding test statistics might not be straightforward.

\begin{enumerate} 
\setcounter{enumi}{3}
\item \textbf{Marginal LOCO parameter:} We introduce a novel parameter, called marginal LOCO, defined after selecting a model. In its plain version (without added randomization), the test statistic for marginal LOCO parameter is not asymptotically Gaussian so we add randomization to satisfy this requirement and apply our tools to get valid inference for this parameter after selection. The marginal LOCO parameter we introduce is inspired by the Leave One Covariate (LOCO) parameter of \cite{rinaldo2016bootstrapping}. These two are, however, different as explained in Section \ref{sec:loco}.

\end{enumerate}

We emphasize that this work analyzes three types of randomizations. 
\begin{enumerate}[label=(\alph*)]
\item \label{item:a:cv:random} Randomizing the vector consisting of quality values, e.g.~cross-validation error vector or the vector consisting of AIC criteria evaluated across different models. This randomization enables asymptotic normality of the corresponding randomized vector, a crucial assumption needed for the post-selection validity of our constructed test statistics.
\item Randomizing the test statistic for the marginal LOCO parameter is essentially done for the same reason as \ref{item:a:cv:random}. We state this separately since here randomization is only applied to the test statistic used for inference and does not modify the selection event.
\item Randomizing the objective function as in randomized Lasso increases statistical power by leaving more information for inference, leading to shorter confidence intervals. We analyze this procedure in conjunction with \ref{item:a:cv:random} (Section \ref{sec:randomized:lasso:cv}).
\end{enumerate}

\subsection{Outline} In Section \ref{sec:gen:framework}, we present a general framework for selective inference. We apply this framework to a wide range of examples, starting with inference after running Lasso with a data-independent penalty level $\lambda$ fixed in advance, in Section \ref{sec:lasso:fixed:lambda}. In Section \ref{sec:nonrandomized:lasso:cv}, we present a way to do valid inference after running cross-validated Lasso. In Section \ref{sec:randomized:lasso:cv}, we show how to account for cross-validation after running randomized model selection procedures to achieve greater power. Two additional applications, inference for the selected coefficients after model selection under AIC and inference for the marginal LOCO parameter, are presented in Section \ref{sec:AIC} and Section \ref{sec:loco}, respectively.

\section{General framework for selective inference} \label{sec:gen:framework}

Before going into details for each of the specific examples let us describe the general framework we propose. 
Suppose our dataset is $S\sim\mathbb{F}_n$, where $\mathbb{F}_n$ is a data generating distribution. We make further assumptions on $\mathbb{F}_n$ in specific examples in Section \ref{sec:nonrandomized:lasso:cv}. We run a model selection procedure $\mathcal{M}$ on data $S$. We assume that the selected model $\mathcal{M}(S)=M$ depends on $S$ only through a data vector $\widetilde{D}=\widetilde{D}(S,M)$ (we intentionally save $D$ for future use). For example, in Lasso example with fixed $\lambda$, $\mathcal{M}$ becomes the Lasso objective and $M=(E,s_E)$ becomes the observed selected model. $\widetilde{D}$ is mathematically complicated so we leave the details to Section \ref{sec:lasso:fixed:lambda}. This setting is fairly general and applicable for model selection procedures other than the Lasso as illustrated in Sections \ref{sec:AIC} and \ref{sec:loco}.

There are three important objects in our framework. 

\begin{itemize} 

\item  The \textit{parameter of interest} $\theta=\theta(\mathbb{F}_n, M)$ is a function of $\mathbb{F}_n$ and it is chosen after we observe the selected model $M$.
 
  \textbf{Example:} In regression examples with $S=(X,y)\in\mathbb{R}^{n\times p}\times\mathbb{R}^n$ we might chose $\theta$ to be the population regression coefficient $\beta^*_E=\left(\mathbb{E}_{\mathbb{F}_n}\left[X_E^\top X_E\right]\right)^{-1}\mathbb{E}_{\mathbb{F}_n}\left[X_E^\top y\right]$ corresponding to the set $E$ of selected predictors, a result of a model selection procedure. $X_E$ are the columns in $X$ corresponding to set $E$ and $X_{-E}$ are the rest of the columns of $X$.
   However, $\theta$ can be other parameter as well. In Section \ref{sec:loco}, $\theta$ is chosen to measure how much a single selected predictor affects the second moment of the population residuals. We call this $\theta$ the {\it Marginal LOCO parameter}.

\item The corresponding \textit{test statistic} $T=T(S,M)$, also called the \textit{target statistic}, centered around $\theta$ pre-selection (as in our example above, suppose we treat $E$ as fixed and not chosen in a data-dependent way). In other words, under $\mathbb{F}_n$ and treating $E$ as non-random, $T-\theta$ follows a mean-zero Gaussian distribution asymptotically and we use $T$ for inference on $\theta$ ignoring selection. Taking selection into consideration, $T$ is generally not Gaussian.

\textbf{Example:} When we do inference on the population regression coefficient $\theta=\beta_E^*$, we take the target statistic to be $T=\bar{\beta}_E=\left(X_E^\top X_E\right)^{-1}X_E^\top y$, the ordinary least square (OLS) estimator calculated on $(X_E, y)$. Under mild assumptions, $\bar{\beta}_E-\beta_E^*$ is asymptotically normal under original data generating distribution treating $E$ as fixed in advance.

\item  The \textit{selection event} $\mathcal{S}=\left\{S':\mathcal{M}(S')=M\right\}$ describes all possible data sets $S'$ for which running the same model selection procedure $\mathcal{M}$ gives the identical observed model $M$ as on the original dataset. Since $\mathcal{M}$ depends on $S$ through $\widetilde{D}$ only, we can write the selection event only in terms of $\widetilde{D}$, denoting a different parametrization of $\mathcal{S}$ as $\mathcal{S}_{\widetilde{D}}$.
Throughout this work, we only consider the selection events $\mathcal{S}_{\widetilde{D}}$ that are affine in $\widetilde{D}$, i.e.
\begin{equation} \label{eq:affine:selection}
	\widetilde{A} \widetilde{D} \\ \leq \tilde{a}_n,	
\end{equation}
where we assume $\tilde{a}_n\rightarrow\tilde{a}$ as $n\rightarrow\infty$. In general, $\widetilde{A}$ is a matrix, $\widetilde{D}$ and $\tilde{a}_n$ are vectors and the inequality is coordinate-wise. We will write them out explicitly in each example later.

\textbf{Example:} The constraints in $\mathcal{S}$ might include the model selection adjustment coming from looking at the minimizer of the vector of models' quality values. Specifically, $\widetilde{D}$ includes the randomized version of this vector, denoted as $Err_R$. Conditioning on observing $r^*$, the index of the minimizer of $Err_R$, induces an affine constraint on $Err_R$. This constraint becomes part of $\mathcal{S}_{\widetilde{D}}$ along with other constraints coming from additional procedures we run on the data.
  
\end{itemize}

Since the parameter $\theta$ has been chosen after looking at the outcome $M$, we need inference on $\theta$ based on the distribution of $T$ conditional on $\mathcal{S}$. The conditioning ``adjusts'' for the pre-selection asymptotic Gaussian distribution of $T-\theta$ to provide a valid post-selection distribution, which we use for inference on $\theta$. In order to have the post-selection distribution of $T-\theta$ not depending on nuisance parameters other than $\theta$, we need to have the selection event only written in terms of $T$. Thus we decompose $\widetilde{D}$ in terms of $T$, and $\mathcal{S}_{\widetilde{D}}$ is now described through $T$ instead of $\widetilde{D}$.  

We describe the re-parametrization of $\mathcal{S}_{\widetilde{D}}$ in terms of $T$ and randomization $\omega$. Assuming that jointly $(T,\widetilde{D})$ is an asymptotically multivariate Gaussian vector, we decompose $\widetilde{D}=\Sigma_{\widetilde{D},T}\Sigma_T^{-1}T+N_{\widetilde{D}}$, where $\Sigma_{\widetilde{D},T}$ and $\Sigma_T$ are the corresponding covariance matrices, $N_{\widetilde{D}}$ is a vector independent of $T$ and we condition on it later. This allows us to write $\mathcal{S}_{\widetilde{D}}$ in terms of $T$ and we denote this new parametrization as $\mathcal{S}_T$. Using the affine representation of $\mathcal{S}_{\widetilde{D}}$, we write $\mathcal{S}_T$ as
\begin{equation*}
	\widetilde{A} \left(\Sigma_{\widetilde{D},T}\Sigma_T^{-1}T+N_{\widetilde{D}}\right) \leq\tilde{a}_n.
\end{equation*}
Given that $T$ is asymptotically Gaussian pre-selection, we derive its asymptotic post-selective distribution by conditioning this pre-selection Gaussian that lands on the set $\mathcal{S}_T$.  This general framework will be used throughout the paper. In specific examples, we elaborate what $\mathcal{S}$, $\mathcal{S}_{\widetilde{D}}$, $\mathcal{S}_T$ are and how to handle additional randomization. 

We emphasize notations for four different distributions which we use frequently in this paper.

\begin{itemize} 

\item $\mathbb{F}_n$: the distribution of the data $S$ \textit{pre-selection}. Since $T$ and $\widetilde{D}$ are functions of $S$ and $M$, to save notations, we also use $\mathbb{F}_n$ to denote the pre-selection distribution of $(T,\widetilde{D})=(T(S,M), \widetilde{D}(S,M))$, treating $M$ as fixed. Index $n$ denotes that the underlying data generating distribution can change with $n$. $\mathbb{F}_n^n$ denotes the distribution of $n$ i.i.d.~copies from $\mathbb{F}_n$; at places with no ambiguity we use $\mathbb{F}_n$ instead of $\mathbb{F}_n^n$.
 
\item $\mathbb{F}_n^*$: the distribution of the data $S$ \textit{post-selection}, i.e.~the distribution of the data $S\sim\mathbb{F}_n$ conditional on the selection event $\mathcal{S}$. $\mathbb{F}_n^*$ is also used for distribution of $(T,\widetilde{D})\sim\mathbb{F}_n$ conditional on the selection event $\widetilde{D}\in\mathcal{S}_{\widetilde{D}}$.

\item $\Phi$: the {\it asymptotic} Gaussian distribution of $(T,\widetilde{D})$ \textit{pre-selection}. We assume that under $S\sim\mathbb{F}_n$ and non-random $M$
\begin{equation} \label{eq:D:T:gaussian}
	\begin{pmatrix}	T(S,M) \\ \widetilde{D}(S,M) \end{pmatrix} \overset{d}{\rightarrow}
	 	 \mathcal{N}\left(\begin{pmatrix}
		\theta \\ \mu_{\widetilde{D}} \end{pmatrix}, \begin{pmatrix} \Sigma_T & \Sigma_{T,\widetilde{D}} \\ \Sigma_{\widetilde{D},T} & \Sigma_{\widetilde{D}} \end{pmatrix}\right) \sim \Phi
\end{equation}
as $n\rightarrow\infty$ for some covariance matrices $\Sigma_T$, $\Sigma_{\widetilde{D}}$ and cross-covariance matrix $\Sigma_{\widetilde{D},T}=\Sigma_{T,\widetilde{D}}^\top$. Since we require the above convergence to hold, we assume $\textnormal{dim}(T)+\textnormal{dim}(\widetilde{D})$ does not grow with the sample size $n$. In regression examples, it implies the number of predictors $p$ to be fixed.

\item $\Phi^*$: the \textit{post-selection} counterpart of $\Phi$. More specifically, for $(Z_T,Z_{\widetilde{D}})\sim\Phi$
\begin{equation*}
 \begin{pmatrix} Z_T \\ Z_{\widetilde{D}} \end{pmatrix}\:\Big|\:\left\{\widetilde{A}Z_{\widetilde{D}}\leq \tilde{a}\right\} \sim \Phi^*.
\end{equation*}
\end{itemize}

\subsection{Selective pivot} \label{sec:selective:pivot}
The joint asymptotic normality of $(T,\widetilde{D})$ is sufficient for us to construct a test statistic valid post-selection, which we call the \textit{selective pivot}. In particular, it is defined as 
\begin{equation} \label{eq:nonrandomized:pivot}
\begin{aligned} 
	& \mathcal{P}\left((T,\widetilde{D}); \widetilde{A}, \tilde{a}_n\right) \\
	&= \mathbb{P}_{(Z_T, Z_{\widetilde{D}})\sim\Phi}\left\{\left\|Z_T-\theta\right\|_2\leq \left\|T-\theta\right\|_2\:\Big|\: Z_T\in\mathcal{S}_T, Z_{\widetilde{D}}-\Sigma_{\widetilde{D},T}\Sigma_T^{-1}Z_T=N_{\widetilde{D}} \right\} \cr
	&=\mathbb{P}_{(Z_T, Z_{\widetilde{D}})\sim\Phi}\left\{\left\|Z_T-\theta\right\|_2\leq \left\|T-\theta\right\|_2\:\Big|\: Z_{\widetilde{D}}\in\mathcal{S}_{\widetilde{D}}, Z_{\widetilde{D}}-\Sigma_{\widetilde{D},T}\Sigma_T^{-1}Z_T=N_{\widetilde{D}}\right\},
\end{aligned}
\end{equation}
where $N_{\widetilde{D}}=\widetilde{D}-\Sigma_{\widetilde{D},T}\Sigma_T^{-1}T$. We define a similar quantity for inference after randomized model selection procedures in Section \ref{sec:randomized:lasso:cv}. Note that the probabilities on the RHS above are only with respect to $(Z_T, Z_{\widetilde{D}})\sim\Phi$. 
Since we condition on $Z_{\widetilde{D}}-\Sigma_{\widetilde{D},T}\Sigma_T^{-1}T$, the pivot depends only on $T$, $\theta$ and the covariance matrices, but not on $\mu_{\widetilde{D}}$. This makes the selective pivot a valid test statistic for $\theta$.

\begin{remarks}
\begin{itemize}[leftmargin=*]

\item[--] It is clear that if $(T,\widetilde{D})$ were exactly Gaussian, the selective pivot is uniformly distributed. More precisely, assuming $\mathbb{F}_n=\Phi$, under $\mathbb{F}_n^*$ the following distributional result holds
\begin{equation*}
	\mathcal{P}\left((T,\widetilde{D}); \widetilde{A}, \tilde{a}_n\right)\sim\textnormal{Unif}(0,1).
\end{equation*} 

\item[--] When $T$ is one-dimensional, the pivot above becomes the truncated Gaussian (TG) test statistic of \cite{lee2013exact}. 

\item[--] Conditioning on the observed value $N_{\widetilde{D}}$ is crucial in removing the dependence of the selective pivot on the nuisance parameter $\mu_{\widetilde{D}}$. Without conditioning we would have the quantity 
\begin{equation} \label{eq:alt:pivot}
	\mathbb{P}_{(Z_T, Z_{\widetilde{D}})\sim\Phi}\left\{\left\|Z_T-\theta\right\|_2\leq \left\|T-\theta\right\|_2\:\Big|\: Z_{\widetilde{D}}\in\mathcal{S}_{\widetilde{D}}\right\}
\end{equation}
 to depend on both $\theta$ and $\mu_{\widetilde{D}}$, so it is not a valid test statistic for doing inference on $\theta$. 
Constructing a test statistic by plugging in the estimates for the nuisance parameter in \eqref{eq:alt:pivot} does not lead to valid inference for $\theta$, as the resulting test statistic does not have estimable cumulative distribution function \citep{leeb2006can, leeb2006performance}. 
\end{itemize}
\end{remarks}

Without assuming the exact normality on $(T,\widetilde{D})$, the following theorem proves $\mathcal{P}$ is asymptotically pivotal after selection given \eqref{eq:D:T:gaussian} holds, i.e.~assuming $(T,\widetilde{D})$ satisfy the Central Limit Theorem (CLT) pre-selection. The asymptotic convergence is under the conditional distribution $\mathbb{F}_n^*$, implying conditional validity of the proposed test statistic post-selection.
The proof of the theorem is given in Section \ref{sec:proofs} in the appendix.

\begin{theorem}[Valid selective pivot after model selection] \label{thm:nonrandomized:pivot} Assuming \eqref{eq:affine:selection} and \eqref{eq:D:T:gaussian} hold, we have that under $(T,\widetilde{D})=(T(S,M),\widetilde{D}(S,M))\sim\mathbb{F}_n^*$
	\begin{equation*}
		\mathcal{P}\left((T,\widetilde{D}); \widetilde{A},\tilde{a}_n\right)\overset{d}{\rightarrow}\textnormal{Unif}(0,1)
	\end{equation*}
	as $n\rightarrow\infty$. In other words, under $(T,\widetilde{D})=(T(S,M), \widetilde{D}(S,M))\sim\mathbb{F}_n$ we have
	\begin{equation*}
		\mathbb{P}_{(Z_T, Z_{\widetilde{D}})\sim\Phi}\left\{\left\|Z_T-\theta\right\|_2\leq \left\|T-\theta\right\|_2\Big|\: Z_{\widetilde{D}}\in\mathcal{S}_{\widetilde{D}}, Z_{\widetilde{D}}-\Sigma_{\widetilde{D},T}\Sigma_T^{-1}T =N_{\widetilde{D}}\right\} \Big |\:\widetilde{D}\in\mathcal{S}_{\widetilde{D}} \overset{d}{\rightarrow} \textnormal{Unif}(0,1)
	\end{equation*}
	as $n\rightarrow\infty$.
\end{theorem}

\subsection{Post-selection consistency}

In order to make the selective pivot useful in practice, we need to use estimated covariance matrices. As we consider low-dimensional examples in this paper, the estimated covariance matrices are consistent pre-selection. The following Lemma shows that, under the conditions of Theorem \ref{thm:nonrandomized:pivot}, these estimated covariances are also consistent post-selection.

\begin{lemma}[Post-selection consistency] \label{lemma:consistency}
	Given a parameter $\xi=\xi(\mathbb{F}_n,M)$ we assume there exists an estimator $\hat{\xi}=\hat{\xi}(S,M)$ such that under $S\sim\mathbb{F}_n$ and treating $M$ as fixed, we have the following convergence in probability
	\begin{equation*}
		\hat{\xi}\overset{P}{\rightarrow}\xi
	\end{equation*}
	as $n\rightarrow\infty$, i.e.~for every $\epsilon>0$ we have $\underset{n\rightarrow\infty}{\lim}\mathbb{F}_n\left\{\left\|\hat{\xi}-\xi\right\|_2\geq \epsilon\right\}=0$. Assuming \eqref{eq:affine:selection} and \eqref{eq:D:T:gaussian} hold,  under $S\sim\mathbb{F}_n^*$, i.e.~post-selection (conditional on selection), we also have
	\begin{equation*}
	 \hat{\xi}\overset{P}{\rightarrow} \xi	
	\end{equation*}
	as $n\rightarrow\infty$, i.e.~for every $\epsilon>0$ we have $\underset{n\rightarrow\infty}{\lim}\mathbb{F}^*_n\left\{\left\|\hat{\xi}-\xi\right\|_2\geq \epsilon\right\}=0$.
\end{lemma}


\section{Inference after Lasso with random $X$ and $\lambda$ pre-fixed} \label{sec:lasso:fixed:lambda}

We apply our framework to the problem of doing inference after Lasso when the design matrix $X$ is random and $\lambda$ is fixed in advance. We focus on the $\ell_2$ loss; however, our technique transfers to smooth convex losses. Recall that, with any fixed $\lambda$ in \eqref{eq:lasso:objective:simple}, $E$ is the set of nonzero coefficients in the Lasso solution $\hat{\beta}$, and $s_E$ contains their signs.

We start by describing the data vector and selection event for the Lasso. The selection event of interest is $\mathcal{S}_{E,s_E}$ from \eqref{eq:lasso:selection:fixed:lambda}, consisting of the event that Lasso selected predictors in $E$ together with the signs of the estimated predictors being fixed at $s_E$.
We write this selection event in terms of Karush-Kuhn-Tucker (KKT) conditions of Lasso for clearer expression.  To do that, let us define the so-called \textit{data vector} $D$ as 
\begin{equation}\label{eq:D:vector}
	D=D\left((X,y),E\right)=\begin{pmatrix}
		\bar{\beta}_E \\ X_{-E}^\top \left(y-X_E\bar{\beta}_E\right)
	\end{pmatrix}.
\end{equation} 
Recall that, $\bar{\beta}_E$ is the OLS estimator of regressing $y$ on $X_E$.

\begin{remark}
$D$ is a special case of $\widetilde{D}$ we defined earlier. In other words, $D$ serves as the data vector for Lasso with fixed $\lambda$ and we later use $\widetilde{D}$ to represent the data vector for Lasso where optimal $\lambda$ is chosen via cross-validation (cross-validated Lasso).
\end{remark}

\noindent With some algebraic calculations and a pre-fixed $\lambda$ in \eqref{eq:lasso:objective:simple}, the KKT conditions (hence selection event) can be represented in terms of $D$ as
\begin{equation} \label{eq:lasso:affine:constraints}
	\begin{pmatrix}
		-\textnormal{diag}(s_E) & 0 \\ 0 &  I_{p-|E|} \\ 0 &  -I_{p-|E|} 
	\end{pmatrix} D \leq \begin{pmatrix}
		-\lambda \textnormal{diag}(s_E)\left(X_E^\top X_E\right)^{-1}s_E \\ \lambda 1_{p-|E|}-\lambda X_{-E}^\top \left(X_E^\top \right)^{\dagger}s_E \\ \lambda 1_{p-|E|}+\lambda X_{-E}^\top \left(X_E^\top \right)^{\dagger}s_E
	\end{pmatrix},
\end{equation}
where $(X_E^\top)^{\dagger}=X_E\left(X_E^\top X_E\right)^{-1}$, $I_{p-|E|}$ is the identity matrix of dimension $p-|E|$ and $1_{p-|E|}\in\mathbb{R}^{p-|E|}$ is a vector of all ones \cite[Theorem 4.3]{lee2013exact}. To ease our notation, let us write the selection event $\mathcal{S}$ in terms of $D$ as 
\begin{equation*}
	D\in\mathcal{S}_D=\left\{D'\in\mathbb{R}^p: A\cdot D'\leq a_n\right\},
\end{equation*} 
with $A$ and $a_n$ defined accordingly as in the above inequality \eqref{eq:lasso:affine:constraints}. Properly scaled $a_n$  converges to a fixed vector (usually the case by SLLN).

After observing the selected set $E$, an analyst might decide to do inference for the parameter $\beta_E^*$, the regression population coefficients corresponding to the selected model. In that case the target statistic we use is $\bar{\beta}_E$ that is asymptotically normal pre-selection with mean $\beta_E^*$ under mild moment conditions. In practice, with some prior knowledge, an analyst might decide to do inference on another population regression coefficients $\beta_{\widetilde{E}}^*$ corresponding to set $\widetilde{E}$ that is not necessarily equal to $E$. Our framework provides valid inference for $\beta_{\widetilde{E}}^*$ in this case as well.

For inference on $\beta_E^*$ via $\bar{\beta}_E$, \eqref{eq:D:T:gaussian} holds under mild moment assumptions on the data generating mechanism $\mathbb{F}_n$ in the random $X$ and fixed $p$ setting. Thus we can decompose $D$ in terms of $T$ as $D=\Sigma_{D,T}\Sigma_T^{-1}T+N_D$, where $N_D$ is asymptotically independent of $T$. Hence, by conditioning on $N_D$ (fixing $N_D$ at its observed value), we can rewrite the selection event $\mathcal{S}$ in terms of $T$ as 
\begin{equation} \label{eq:lasso:constraints:target}
	T\in\mathcal{S}_T=\left\{T':A\cdot\Sigma_{D,T}\Sigma_T^{-1}T'\leq a_n - AN_D\right\}.
\end{equation}
Notice that, $\mathcal{S}$, $\mathcal{S}_D$ represent the same selection event with different parameterizations. $\mathcal{S}_T$ is $\mathcal{S}_D$ conditioning on one more variable $N_D$. Constructing the selective pivot based on \eqref{eq:nonrandomized:pivot} gives valid inference for our target parameter $\beta_E^*$. 

In practice, we estimate the covariance matrices using pairs-bootstrap. In the random $X$ and  fixed $p$ setting, the estimates via pairs bootstrap are consistent pre-selection \citep{freedman1981bootstrapping, buja2014conspiracy}; using Lemma \ref{lemma:consistency} gives us that these estimates are consistent post-selection as well.

\section{Inference after cross-validated Lasso} \label{sec:nonrandomized:lasso:cv}

In this section, we present a way to do inference after Lasso where the penalty level $\lambda=\lambda^{cv}$ has been chosen using cross-validation.  

\subsection{$K$-fold cross-validation}

Our goal is to extend the ideas from the previous section so that we can do valid inference after $\lambda$ has been chosen via cross-validation. Let us first review the cross-validation procedure. Given data $(X,y)\in\mathbb{R}^{n\times p}\times\mathbb{R}^n$, we split it  into $K$ disjoint folds denoted as $(X^k,y^k)\in\mathbb{R}^{n_k\times p}\times\mathbb{R}^{n_k}$, $k=1,\ldots, K$, containing $n_1,\ldots, n_K$ observations, respectively. The data without the $k$-th fold is denoted as $(X^{-k}, y^{-k})$. Suppose we choose a grid of $\lambda$'s to be $\Lambda = \{\lambda_1,\ldots, \lambda_L\}$. To choose $\lambda^{cv}$, for each $\lambda \in \Lambda$, we follow the steps below: 

\begin{enumerate}

\item For each fold $k$ and each $\lambda\in\Lambda$, we compute the Lasso estimator $\hat{\beta}^{-k}(\lambda)$ on the training data $(X^{-k},y^{-k})$: 
\begin{equation*}
	\hat{\beta}^{-k}(\lambda) = \textnormal{arg}\underset{\beta\in\mathbb{R}^p}{\min}\:\frac{1}{2}\left\|y^{-k}-X^{-k}\beta\right\|_2^2+\lambda\|\beta\|_1. 
\end{equation*} 

\item For each fold $k$ and each $\lambda\in\Lambda$, we evaluate the error of this estimator on the test data $(X^k,y^k)$ as 
	\begin{equation*}
		Err(\lambda,k) = \frac{1}{n_k}\left\|y^k-X^k\hat{\beta}^{-k}(\lambda)\right\|_2^2.
	\end{equation*}

\item We define the cross-validated error (CV error from now on) for each $\lambda\in\Lambda$ as 
	\begin{equation*}
		Err(\lambda) = \sum_{k=1}^K Err(\lambda, k).
	\end{equation*}
	 We use phrase \textit{CV curve} to denote the vector 
	\begin{equation*}
		Err=\left(Err(\lambda_1),\ldots, Err(\lambda_L)\right).
	\end{equation*}

\item We pick $\lambda^{cv}=\lambda_{l^*}$, $1\leq l^*\leq L$ ($l^*$ is the index of the minimizer), that minimizes the CV curve: 
\begin{equation} \label{eq:cv:nonrandom}
	\lambda^{cv} = \textnormal{arg}\underset{\lambda\in\Lambda}{\min}\: Err(\lambda).
\end{equation}

\end{enumerate}

Let us see what happens when we use the truncated Gaussian (TG) test statistic of \citealt{lee2013exact} without adjusting for CV for the selected coefficients, output by cross-validated Lasso. With simulation settings described in its caption, Figure \ref{fig:lee_et_al} shows a clear violation of $p$-values from the uniform distribution (the straight 45 degrees line). We see that in this case, accounting for model selection via Lasso using TG is not enough since we do not take into account the fact that $\lambda$ has also been chosen in a data-dependent manner. The naive $p$-values are also added for comparison. They are constructed based on the normal quantiles by using asymptotic normality of $T$ pre-selection. Since they ignore both model selection and cross-validation, the naive $p$-values deviate further from uniform.

\begin{figure}[h!]
  \centering
    \includegraphics[width=0.6\textwidth]{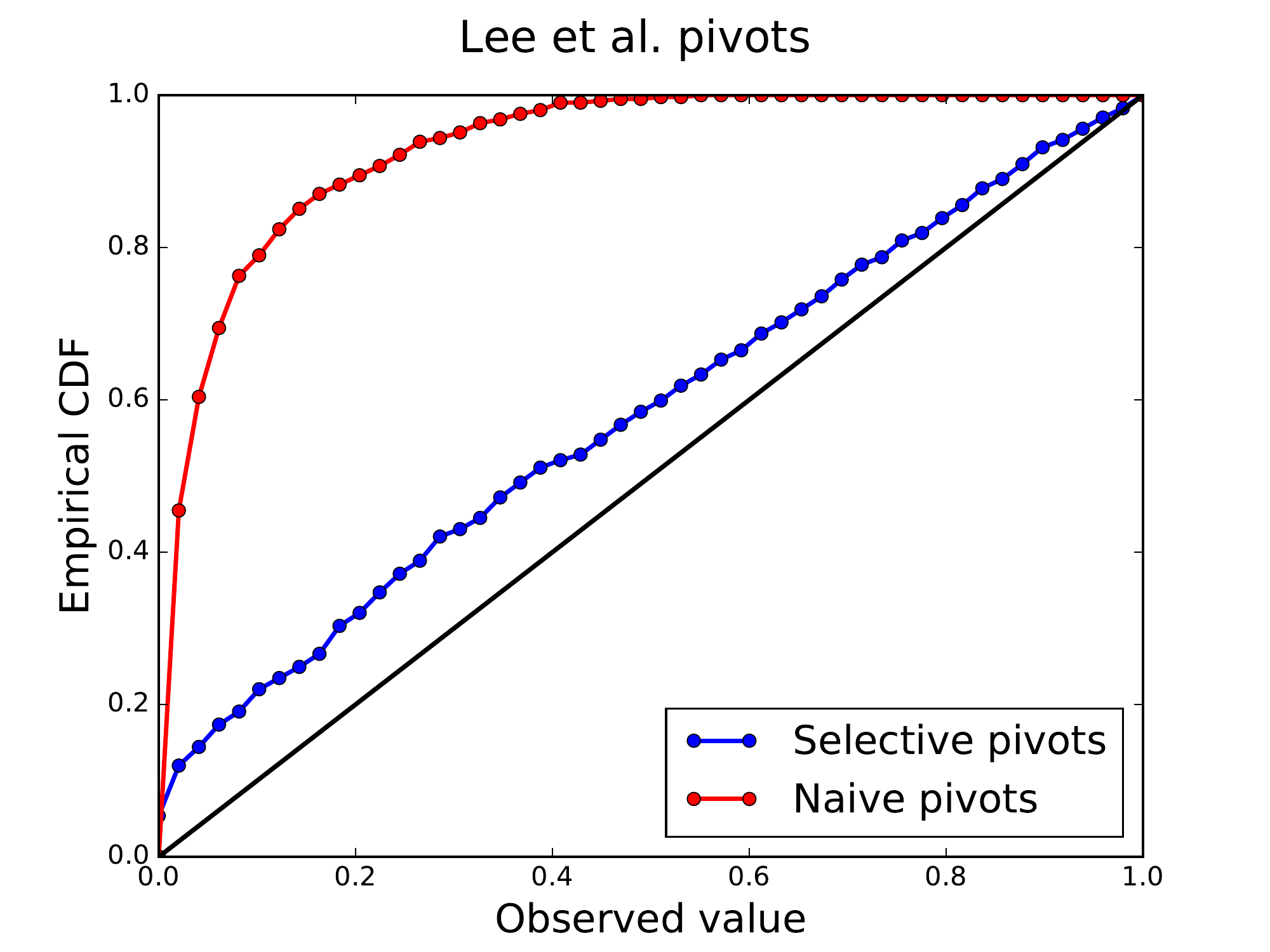}
    \caption{The $p$-values constructed for the selected coefficients using TG test statistic of \citealt{lee2013exact} (blue) and the naive ones (red) after running Lasso to select the model with $\lambda$ chosen using cross-validation. We take $n=500, p=100$ and the entries of matrix $X$ are generated as independent standard normal random variables, with the columns of $X$ normalized to have empirical variance 1. The response vector is $y\sim\mathcal{N}(0,I_n)$, i.e.~a null signal setting.}
    \label{fig:lee_et_al}
\end{figure}

If $\lambda$ has been chosen in a data-dependent manner, the selection event we ``look'' at differs from the selection event with fixed $\lambda$; it becomes much more complicated to describe.

\subsection{Randomized CV curve} \label{sec:randomized:cv:curve}

To account for cross-validation, it is necessary that we adjust the selection region $\mathcal{S}$ further taking into account the minimizer of the CV curve as in \eqref{eq:cv:nonrandom}. In other words, the right selection event $\mathcal{S}$ is composed of two parts: the selection effect from Lasso described in \eqref{eq:lasso:affine:constraints} and the selection resulted from CV. By conditioning on $l^*$, the index of the cross-validation minimizer, we can rewrite the CV part of the selection event $Err\in\{Err'\in\mathbb{R}^L:Err'_{l^*}=\underset{1\leq l\leq L}{\min}Err'_l\}$ as 
\begin{equation} \label{eq:cv:affine}
	B_{l^*} \cdot Err\leq 0,
\end{equation}
with $B_{l^*}\in\mathbb{R}^{(L-1)\times L}$ a matrix of zeros, ones and minus ones. Intuitively, if $(Err,T,D)$ is a jointly Gaussian vector, by decomposing $Err$ in terms of $T$ as $Err=\Sigma_{Err,T}\Sigma_T^{-1}T+N_{Err}$, we combine constraints \eqref{eq:lasso:constraints:target} and \eqref{eq:cv:affine}, and thus the combined selection event will be affine in $T$. The selective pivot construction will follow from Section \ref{sec:selective:pivot}. However, $Err$ alone may not follow a multivariate Gaussian, and thus we have no hope in $(Err,T,D)$.

To see why $Err$ is not multivariate Gaussian, suppose the linear model $y=X_{E^*}\beta_{E^*}+\epsilon$, $\epsilon|X\sim\mathcal{N}(0,I_n)$, is true for a set $E^*\subset\{1,\ldots,p\}$. For some $\lambda$ in the grid $\Lambda$ which selects a set $E\supseteq E^*$, according to our cross-validation procedure, one fold in the cross-validated error centered around its expected value looks like: 
\begin{equation*}
\begin{aligned}
  &\frac{1}{n_k}\left(\left\|y^k-X^k_E\bar{\beta}^{-k}(\lambda)\right\|^2_2-\mathbb{E}\left[\left\|y^k-X^k_E\bar{\beta}^{-k}(\lambda)\right\|^2_2\right]\right) \\
  &=\frac{1}{n_k}\left(\left\|\epsilon^{k}-X^k_E({X_E^{-k}}^{\top}X_E^{-k})^{-1} {X_E^{-k}}^{\top}\epsilon^{-k}\right\|^2_2 
  -\mathbb{E}\left[\left\|y^k-X^k_E\bar{\beta}^{-k}(\lambda)\right\|^2_2\right]\right) \\
  &=\frac{\left\|\epsilon^{k}\right\|_2^2-\mathbb{E}\left[\left\|\epsilon^k\right\|_2^2\right]}{n_k}
   -2\frac{{\epsilon^k}^{\top} X^k_E ({X^{-k}_E}^{\top}X^{-k}_E)^{-1} {X^{-k}_E}^{\top}\epsilon^{-k}}{n_k} +\frac{\left\|X^k_Ev^{-k}\right\|_2^2-\mathbb{E}\left[\left\|X^k_Ev^{-k}\right\|_2^2\right]}{n_k}, 
\end{aligned}
\end{equation*} 
where $\epsilon^k=y^k-X^k_{E^*}\beta_{E^*}\sim\mathcal{N}(0,I_{n_k})$,  $\epsilon^{-k}=y^{-k}-X^{-k}_{E^*}\beta_{E^*}\sim\mathcal{N}(0,I_{n-n_k})$, $v^{-k} = \left({X^{-k}_E}^{\top}X^{-k}_E\right)^{-1}{X^{-k}_E}^{\top}\epsilon^{-k}$. For the purpose of providing intuitive explanation, we assume above that after fitting Lasso on the training data, we compute $\bar{\beta}^{-k}(\lambda)$ as the OLS estimator on $y^{-k}\sim X^{-k}_E$. Using $\hat{\beta}^{-k}(\lambda)$ is more complicated and will make $Err$ deviate further from a Gaussian distribution asymptotically. By the CLT, the first term in the above equation asymptotically follows distribution $Z/\sqrt{n_k}$, where $Z\sim\mathcal{N}(0,1)$. The CLT together with Slutsky's lemma tell us that the second and the third term are of order $O_P(1/n)$, but without a Gaussian distribution. Therefore, we see that marginally, each fold of the cross-validate error and thus a centered and scaled version of $Err(\lambda)$ asymptotically follows $Z/\sqrt{n}$. However, if $\textnormal{supp}(\hat{\beta}^{-k}(\lambda_1))$ and $\textnormal{supp}(\hat{\beta}^{-k}(\lambda_2))$ contain the true set $E^*$ for all $k$, then the first term in the above equation would cancel and only the non-Gaussian part of order $O_P(1/n)$ will be left. As a result, CV curve $Err$ is not asymptotically distributed as multivariate Gaussian, since we see the difference between the two entries $Err(\lambda_1)$ and $Err(\lambda_2)$ is asymptotically non-Gaussian. The left plot in Figure \ref{fig:qq_plots} is an illustration of the non-Gaussianity, and we will come back to it later.

As a remedy, we propose using a randomized version of cross-validation vector, e.g.
\begin{equation} \label{eq:cv:randomizing:residuals}
	Err_R(\lambda,k)=\frac{1}{n_k}\left\|y^k-X^k\hat{\beta}^{-k}(\lambda)\right\|_2^2+\frac{1}{\sqrt{n_k}}R^{k,\lambda},
\end{equation}
where $R^{k,\lambda}\sim\mathcal{N}(0,\tau^2)$, for some pre-specified parameter $\tau$; and $R^{k,\lambda}$ are independent across different $\lambda$ and $k$ and independent of the data. Note the scale of randomization is the same as the scale of the Gaussian random variable in $Err(\lambda)$. Let us denote the joint distribution of the added randomizations $\left\{R^{k,\lambda}:1\leq k\leq K, \lambda\in\Lambda\right\}$ as $\mathbb{F}_R$ and the sum of the randomized cross-validation errors across folds is defined as 
\begin{equation*}
	Err_R(\lambda) = \sum_{k=1}^K Err_R(\lambda, k).
\end{equation*}	
Concatenating $Err_R(\lambda)$, $\lambda\in\Lambda$, into a vector, we call
$Err_R=\left(Err_R(\lambda_1),\ldots, Err_R(\lambda_L)\right)\in\mathbb{R}^{|\Lambda|}$  a \textit{randomized CV error curve}.

According to the above derivation, the added randomization is distributed as Gaussian and is of order $\tau\sim O(1)$. Also, the differences between CV errors across different $\lambda$'s do not cancel out since $R^{k,\lambda}$ are generated independently. Therefore, $Err_R$ is asymptotically jointly Gaussian now. Figure \ref{fig:qq_plots} illustrates this phenomenon for data generated from a null model, i.e.~true model does not contain any variable. The differences between the two coordinates of the randomized CV error curve are much closer to being Gaussian.

\begin{figure}[h!]
  \centering
    \includegraphics[width=0.8\textwidth]{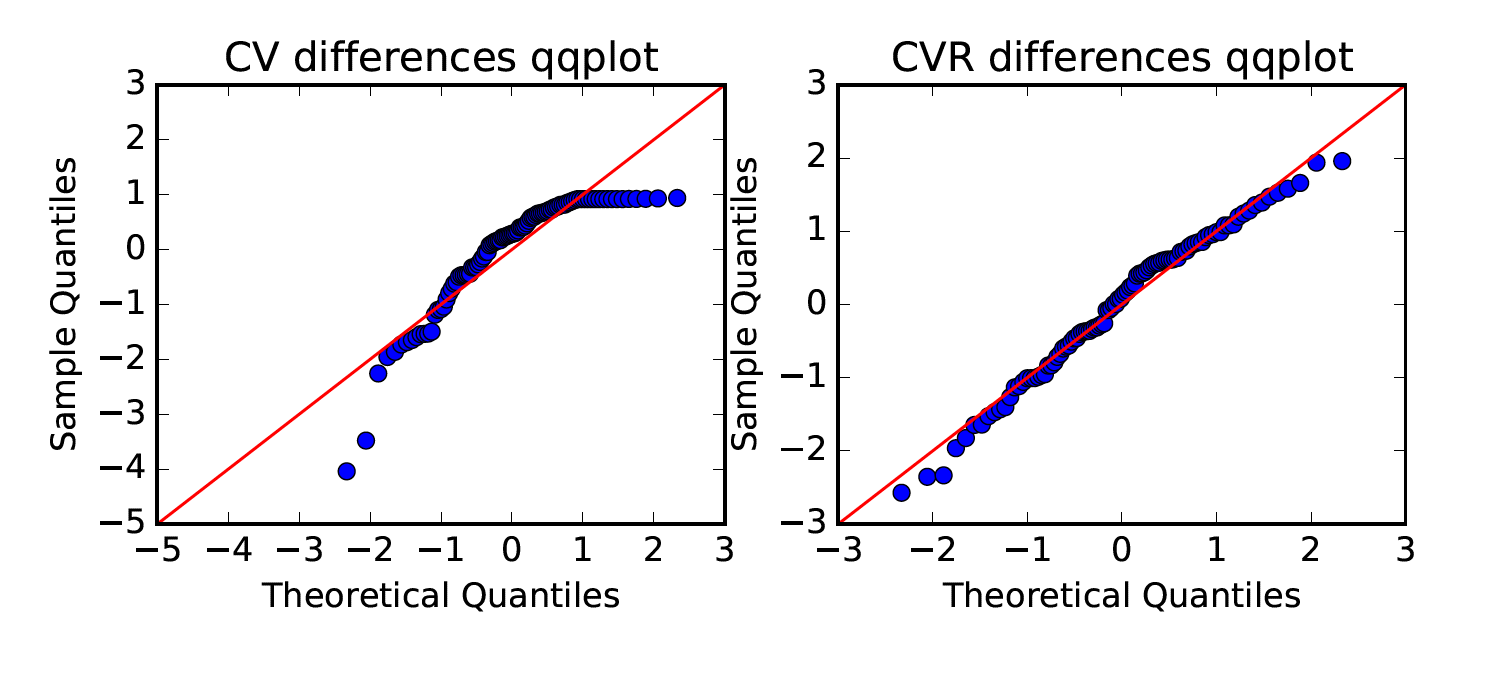}
    \caption{QQ plots of the differences between two coordinates of CV error curve (left) and randomized CV error curve (right), both compared with normal quantiles. }
    \label{fig:qq_plots}
\end{figure}

\begin{remarks}
\begin{itemize}[leftmargin=*]
\item[--]
	In an independent recent work, \citealt{rinaldo2016bootstrapping} proposed a similar randomized construction for establishing a CLT for their LOCO parameter. At the time when cross-validation part of this paper was developed, we were not aware of their results.
\item[--]
	Note that $Err_R$ is unbiased for the CV error. Using the results of \citealt{cv_risk_consistency, dudoit2005asymptotics}, we can also show that $Err_R(\lambda_R^{cv})$ is risk consistent, under similar assumptions. Also, \cite{dudoit2005asymptotics} derived the asymptotic normality of properly scaled and centered cross-validated risk estimator $Err(\lambda^{cv})$, where the centering is around the conditional risk that marginalizes over the validation set and conditions on the training set (assuming we have two folds). However, we need the centered and scaled version of $Err$ to be jointly Gaussian across $\lambda$ values. Furthermore, we want the centering term to be $\mathbb{E}[Err]$, where the expectation marginalizes over the whole dataset.
\end{itemize}
\end{remarks}

Now we choose $\lambda_R^{cv}=\lambda_{r^*}$, $1\leq r^*\leq L$,  where
\begin{equation} \label{eq:cv:random}
	\lambda_R^{cv} = \textnormal{arg}\underset{\lambda\in\Lambda}{\min}\: Err_R(\lambda).
\end{equation}
This $\lambda=\lambda_R^{cv}$ value is then used in the Lasso objective. Recall that cross-validation by itself (without added randomization) is a randomized procedure since it chooses folds randomly. The added randomization is usually of a small order, and it does not affect the value of the minimizer much. On the other hand, it adds enough to the existing variance of $Err(\lambda)$ so that the resulting curve is jointly Gaussian. In the next subsection, we provide rigorous proof of the joint Gaussianity of the data vector and the randomized CV curve.

\subsection{Proving randomized CV curve is asymptotically Gaussian}

We consider the random $X$ and fixed $p$ setting. We start by providing a theorem for general training and testing data. Given this theorem, the main conclusion will follow. We assume data consists of
	 $n$ i.i.d.~observations $(x_i, y_i)\in\mathbb{R}^p\times\mathbb{R}$, $ 1\leq i \leq n$ from distribution $\mathbb{F}$ (does not change with $n$ for simplicity). We split data into disjoint training and test set.
	 Let us denote the training set of size $n_1$ as $(X^{train}, y^{train})\sim\mathbb{F}^{n_1}$ and an independent test set of size $n_2=\Theta(n_1)$ as $(X^{test}, y^{test})\sim\mathbb{F}^{n_2}$. Denote the Lasso estimator $\hat{\beta}^{train}(\lambda)$ we get by solving Lasso objective on $(X^{train}, y^{train})$ for a particular penalty level $\lambda$. 
Given this setting, the assumptions we need are as follows.	 
\begin{itemize}
	\item \textbf{Consistency assumption:} We assume that for each $\lambda\in\Lambda$, the Lasso estimator $\hat{\beta}^{train}(\lambda)$ is consistent for some parameter $\beta_0(\lambda)$ at the rate $n_1^{1/4}$, i.e.~
	\begin{equation*}
		\sqrt{n_1}\left\|\hat{\beta}^{train}(\lambda)-\beta_0(\lambda)\right\|_2^2\overset{P}{\rightarrow} 0
	\end{equation*}
	as $n_1\rightarrow\infty$. Note we assume neither the linear model is true nor there exists any relationship among the parameters $\beta_0(\lambda)$, $\lambda\in\Lambda$.
	
	\item \textbf{Moment assumption:} We further assume the moment conditions under $(x_1,y_1)\sim\mathbb{F}$ as follows:
	\begin{itemize}
	\item $\sigma^2(\lambda)=\textnormal{Var}\left(\left(y_1-x_1^\top\beta_0(\lambda)\right)^2\right)<\infty$,  
	\item $\mathbb{E}\left[\left\|x_1\right\|_2^2(y_1-x_1^\top\beta_0(\lambda))^2\right]<\infty$ and
	\item $\mathbb{E}\left[\left\|x_1\right\|_2^2\right]<\infty$.
	\end{itemize}
\end{itemize}

\begin{theorem}[Normality of the test error] \label{thm:test:error:fixed:p} Suppose the consistency and moment assumptions above hold. Then for all $\lambda\in\Lambda$, we have
\begin{equation*}
	\frac{1}{\sqrt{n_2}}\left(\left\|y^{test}-X^{test}\hat{\beta}^{train}(\lambda)\right\|_2^2 - n_2\mu(\lambda)\right)\overset{d}{\rightarrow}\mathcal{N}\left(0,\sigma^2(\lambda)\right)
\end{equation*}
as $n_1,n_2\rightarrow\infty$, where $\mu(\lambda)=\mathbb{E}\left[\left(y_1-x_1^\top\beta_0(\lambda)\right)^2\right]$.
\end{theorem}
\noindent The proof of the theorem above is given in Section \ref{sec:proofs} in the appendix.

Let us turn to the cross-validation error vector. Recall that for each of the $k=1,\ldots, K$ folds, data is split in two disjoint sets $(X^k, y^k)$ and its complement $(X^{-k}, y^{-k})$ of sizes $n_k=O(n/K)$ and $n-n_k$ respectively. We can apply Theorem \ref{thm:test:error:fixed:p} for each of these $K$ data splits. The conclusion, summarized in the following corollary, implies that the randomized CV error curve is asymptotically jointly Gaussian for all $\lambda$ in the given range.
The proof of the corollary is in Section \ref{sec:proofs} in the appendix.

\begin{corollary}[Randomized CV curve is asymptotically Gaussian] \label{cor:randomized:cv:gaussian}
	Under the consistency and moment assumptions above, we have
	\begin{equation*}
		\sqrt{n}\left(Err_R-\mathbb{E}_{\mathbb{F}\times\mathbb{F}_R}[Err_R]\right) \overset{d}{\rightarrow} \mathcal{N}\left(0,\Sigma_{Err_R}\right)
	\end{equation*}
	as $n\rightarrow\infty$, for some invertible covariance matrix $\Sigma_{Err_R}\in\mathbb{R}^{|\Lambda|\times|\Lambda|}$.
\end{corollary}

\begin{remarks}
\begin{itemize}[leftmargin=*]
\item[--] Here is some intuition on why when doing inference for Lasso with randomized cross-validation, in general, we cannot ignore the adjustment for cross validation and adjust only for the selection event of Lasso. 

When writing constraints coming from cross-validation in the selection event, we look at $\lambda_R^{cv}= \textnormal{arg}\underset{\lambda\in\Lambda}{\min}Err_R(\lambda)$. This event can also be described in terms of the vector of differences $\left(Err_R(\lambda_R^{cv})-Err_R(\lambda):\lambda\in\Lambda\right)$. In other words, it is equivalent to describe the constraints as: $\left\{Err_R(\lambda_R^{cv})-Err_R(\lambda)\leq 0:\lambda\in\Lambda\right\}$. 
Arguments below show that, under the assumptions of this section, $Err_R(\lambda_1)-Err_R(\lambda_2)$ is not in general independent of the data for two given values $\lambda_1,\lambda_2\in\Lambda$ and this implies that the differences of the randomized CV errors $Err(\lambda)$ across $\lambda$ values are not independent of the data.

From the proof of Theorem \ref{thm:test:error:fixed:p}, we have that 
\begin{equation*}
	\frac{1}{\sqrt{n_k}}\left(\left\|y^k-X^k\hat{\beta}^{-k}(\lambda)\right\|_2^2 -n_k\mu(\lambda)\right)=\frac{1}{\sqrt{n}_k}\left(\left\|\epsilon^k(\lambda)\right\|_2^2-n_k\mu(\lambda)\right)+o_P(1),
\end{equation*}	
where $\epsilon^k(\lambda)=\left(\epsilon^k_1(\lambda),\ldots, \epsilon^k_{n_k}(\lambda)\right)=y^k-X^k\beta_0(\lambda)\in\mathbb{R}^{n_k}$. It is worth noting that the random variables $\epsilon^k_i(\lambda)$ are i.i.d.~across $i=1,\ldots, n_k$ and $k=1,\ldots, K$ for each $\lambda$. In general, the parameters $\beta_0(\lambda)$ are not equal across different $\lambda\in\Lambda$ values, thus the random variables $\epsilon^k_i(\lambda)$ are neither equal nor identically distributed across $\lambda\in\Lambda$ for fixed $k=1,\ldots, K$ and $i=1,\ldots, n_k$. For each $\lambda\in\Lambda$, the randomized quantity $Err_R(\lambda)$ becomes asymptotically the scaled and centered sum of i.i.d.~terms $\epsilon_i^k(\lambda)^2$, where the sum is across $i=1,\ldots, n_k$ and $k=1,\ldots, K$, and randomization. Thus the difference $Err_R(\lambda_1)-Err_R(\lambda_2)$ across two values $\lambda_1,\lambda_2\in\Lambda$ is not independent of the data $(X,y)$ as the residual terms $\epsilon_i^k(\lambda_1)^2$ and $\epsilon_i^k(\lambda_2)^2$ do not cancel in general.

\item[--] Under additional assumptions, e.g.~assuming $\beta_0(\lambda)$ are equal across all $\lambda\in\Lambda$ values, we have that the differences $Err_R(\lambda_1)-Err_R(\lambda_2)$ consist of only the added randomization. In this case, the differences between the randomized cross-validation errors $Err_R(\lambda)$ across $\lambda$ values are independent of the data. Therefore, under these strict assumptions it is possible to ignore adjusting for randomized cross-validation. To do valid inference we would adjust only for the Lasso selection event with the penalty chosen based on $\lambda_R^{cv}$. Note that we still choose the penalty level for the Lasso based on the vector of randomized cross-validation errors. Further adjusting for the randomized cross-validation, however, is a more robust approach, requiring less assumptions.
\end{itemize}
\end{remarks}

In order to prove the joint normality of the data vector $D$ and randomized CV curve, we state the following mild assumptions.
\begin{itemize}
	\item $\frac{1}{n}X^\top X\overset{P}{\rightarrow} \mathbb{E}\left[\frac{1}{n}X^\top X\right]$ and $\left(\frac{X^\top X}{n}\right)^{-1}\overset{P}{\rightarrow}\mathbb{E}\left[\left(\frac{X^\top X}{n}\right)^{-1}\right]$ as $n\rightarrow\infty$.
	\item $\mathbb{E}_{(x_1,y_1)\sim\mathbb{F}}\left[\left\|x_1\right\|_2^2\left(y_1-x_1^\top \beta_E^*\right)\right]<\infty$.
\end{itemize}

\begin{corollary}[Joint asymptotic normality of data and CV curve] 
\label{cor:joint:normality} Under the assumptions above, we have that pre-selection 
\begin{equation*}
	\sqrt{n}\left(\begin{pmatrix}	\bar{\beta}_E \\ \frac{1}{n}X_{-E}^\top(y-X_E\bar{\beta}_E) \\ Err_R \end{pmatrix} - \begin{pmatrix}
		\beta_E^* \\ \frac{1}{n}\mathbb{E}\left[X_{-E}^\top\left(y-X_E\beta_E^*\right)\right] \\ \mathbb{E}\left[Err_R\right]
	\end{pmatrix} \right) \overset{d}{\rightarrow}\mathcal{N}\left(0,\Sigma\right)
\end{equation*}
as $n\rightarrow\infty$ for some joint covariance matrix $\Sigma$.
\end{corollary}
The corollary above is proved in Section \ref{sec:proofs} in the appendix.	

\subsection{Inference with randomized cross-validation}

We have shown the joint Gaussianity of the data vector and the randomized CV curve. In this section, we describe the selection event coming from both model selection (Lasso) and randomized cross-validation in detail. Furthermore, we provide pivots that are valid post-selection and thus lead to valid selective inference.

In addition to the constraint coming from the Lasso, the constraint coming from cross-validation is 
\begin{equation} \label{eq:cv:affine:randomized}
	Err_R\in\mathcal{E}_{r^*}=\left\{Err_R'\in\mathbb{R}^L:B_{r^*}\cdot Err_R'\leq 0\right\},
\end{equation}
given $r^*$. Denoting
\begin{equation*}
	\widetilde{A} = \begin{pmatrix}	A & 0 \\ 0 &  B_{r^*}, \end{pmatrix} \;\; \tilde{a}_n=\begin{pmatrix} a_n \\ 0 \end{pmatrix},
\end{equation*}
 the selection event in terms of the joint vector $\widetilde{D}=(D,Err_R)\in\mathbb{R}^{p+L}$ becomes 
\begin{equation*}
	\widetilde{A}\cdot\widetilde{D} \leq \tilde{a}_n. 
\end{equation*}

Recall that in order to do inference for a parameter of interest $\theta$ using a target statistic $T=T\left((X,y),E\right)$ we need to rewrite the selection event in terms of $T$ only. After randomization, as long as vectors $T$, $D$ and $Err_R$ are jointly Gaussian pre-selection, i.e.~under $\left((X,y),R\right)\sim\mathbb{F}_n\times\mathbb{F}_R$ and fixed $E$ 
\begin{equation} \label{eq:cv:joint:asymptotically:normal:vector}
	\begin{pmatrix} T\left((X,y),E\right) \\ D\left((X,y),E\right) \\ Err_R  
	\end{pmatrix} \overset{d}{\rightarrow}
	\mathcal{N}\left(\begin{pmatrix} \theta \\ \mu_D \\ \mu_{Err_R}  \end{pmatrix},
	\begin{pmatrix} 
	\Sigma_T   & \Sigma_{T,D} & \Sigma_{T,Err_R} \\
	\Sigma_{D,T} & \Sigma_D & \Sigma_{D,Err_R} \\ 
	\Sigma_{Err_R, T} & \Sigma_{Err_R, D} & \Sigma_{Err_R} 
	\end{pmatrix} \right)
\end{equation}
as $n\rightarrow\infty$, we can decompose, $Err_R$ and $\widetilde{D}$ in terms of $T$. In order to do valid inference, we need the following two results. The following proposition rewrites the selection region in terms of $T$. Theorem \ref{thm:valid:pivot} justifies that conditioning on the components orthogonal to $T$ and the selection region in terms of $T$ creates an asymptotically valid pivot.

\begin{proposition}[Selection event for Lasso with randomized cross-validation] \label{thm:selevent:Lasso:randomizedCV}
Defining $N_{Err_R}=Err_R-\Sigma_{Err_R,T}\Sigma_T^{-1}T$ and $N_D=D-\Sigma_{D,T}\Sigma_T^{-1}T$, the selection event of Lasso with randomized cross-validation can be described as $T\in\widetilde{S}_T$ for
\begin{equation*}
	\widetilde{\mathcal{S}}_T = \left\{ T'\in\mathbb{R}^{\textnormal{dim}(\theta)}:\begin{pmatrix} A\cdot\Sigma_{D,T}\Sigma_T^{-1}\\ B_{r^*}\cdot\Sigma_{Err_R,T}\Sigma_T^{-1} \end{pmatrix} T'\leq \begin{pmatrix} a_n-AN_D\\ -B_{r^*}\cdot N_{Err_R}\end{pmatrix}\right\}.
\end{equation*}
\end{proposition}

\begin{proof}
 We rewrite the inequality in \eqref{eq:cv:affine:randomized} as 
\begin{equation*}
	B_{r^*}\cdot\left(\Sigma_{Err_R,T}\Sigma_T^{-1}T+N_{Err_R}\right) \leq 0,
\end{equation*}
or equivalently 
\begin{equation} \label{eq:cv:affine:decomposed}
	B_{r^*}\cdot\Sigma_{Err_R,T}\Sigma_T^{-1}T\leq -B_{r^*}N_{Err_R}.
\end{equation}
Combining events \eqref{eq:lasso:affine:constraints} and \eqref{eq:cv:affine:decomposed}, both written in terms of the target statistic $T$, corresponding to the Lasso and cross-validation respectively, we get the form for $\widetilde{S}_T$.
\end{proof}


\begin{theorem}[Pivot valid post-selection for Lasso with randomized cross-validation] \label{thm:valid:pivot}
Assuming \eqref{eq:cv:joint:asymptotically:normal:vector} holds, the selective pivot valid after both cross-validation and model selection is 
\begin{equation*}
	\mathcal{P}\left((T,\widetilde{D});\widetilde{A},\tilde{a}\right)
	= \mathbb{P}_{(Z_T, Z_{\widetilde{D}})\sim\Phi}\left\{\left\|Z_T-\theta\right\|_2\leq \left\|T-\theta\right\|_2\:\Big|\: Z_T\in\widetilde{S}_T, Z_{\widetilde{D}}-\Sigma_{\widetilde{D},T}\Sigma_T^{-1}Z_T=N_{\widetilde{D}} \right\},
\end{equation*}
where $N_{\widetilde{D}}=\begin{pmatrix} N_D \\ N_{Err_R}\end{pmatrix}$ and $\Phi$ corresponds to the Gaussian distribution on the RHS in \eqref{eq:cv:joint:asymptotically:normal:vector}. Furthermore conditional on the selection event,
\[
\mathcal{P}\left((T,\widetilde{D});\widetilde{A},\tilde{a}\right)\sim\text{Unif}\  [0,1].
\]
\end{theorem}

\begin{proof}
Having the selection event written in terms of $T$ and assuming \eqref{eq:cv:joint:asymptotically:normal:vector} holds, we have by Theorem \ref{thm:nonrandomized:pivot} that the pivot above asymptotically follows uniform distribution on $(0,1)$ under the conditional distribution.
\end{proof}

\begin{remarks}
\begin{itemize}[leftmargin=*]

\item[--] Computing the pivots involves estimating the covariance matrices. We use non-parametric covariance estimates via pairs bootstrap throughout the paper except in Section \ref{sec:power} where we use parametric estimates.

\item[--] In practice, there are other ways that we may randomize the cross-validation errors to achieve the joint CLT for $\widetilde{D}$. Another possibility is to randomize within residuals and take 
\begin{equation*}
	Err_R(\lambda,k)=\frac{1}{n_k}\left\|y^k-X^k\hat{\beta}^{-k}(\lambda)+R^{k,\lambda}\right\|_2^2-\frac{1}{n_k} \left\|R^{k,\lambda}\right\|^2_2,
\end{equation*}
for $R^{k,\lambda}\sim\mathcal{N}(0,\tau^2I_{n_k})$. For simplicity, we stick to the additive randomization in \eqref{eq:cv:randomizing:residuals} in this paper.

\item[--] Our inference framework applies to a general loss function used to compute randomized CV error curve: we do not need the loss in \eqref{eq:cv:randomizing:residuals} to be squared error loss or to be the same as in model selection, e.g.~$\ell_2$ loss for Lasso, as long as we have a joint CLT for $(T,D,Err_R)$.
\end{itemize}
\end{remarks}

In Figure \ref{fig:cv_corrected_nonrandomized_lasso_pivots}, we present the selective $p$-values after adjusting for cross-validation using the same data generating mechanism as in Figure \ref{fig:lee_et_al}.
\begin{figure}[h!]
  \centering
    \includegraphics[width=0.60\textwidth]{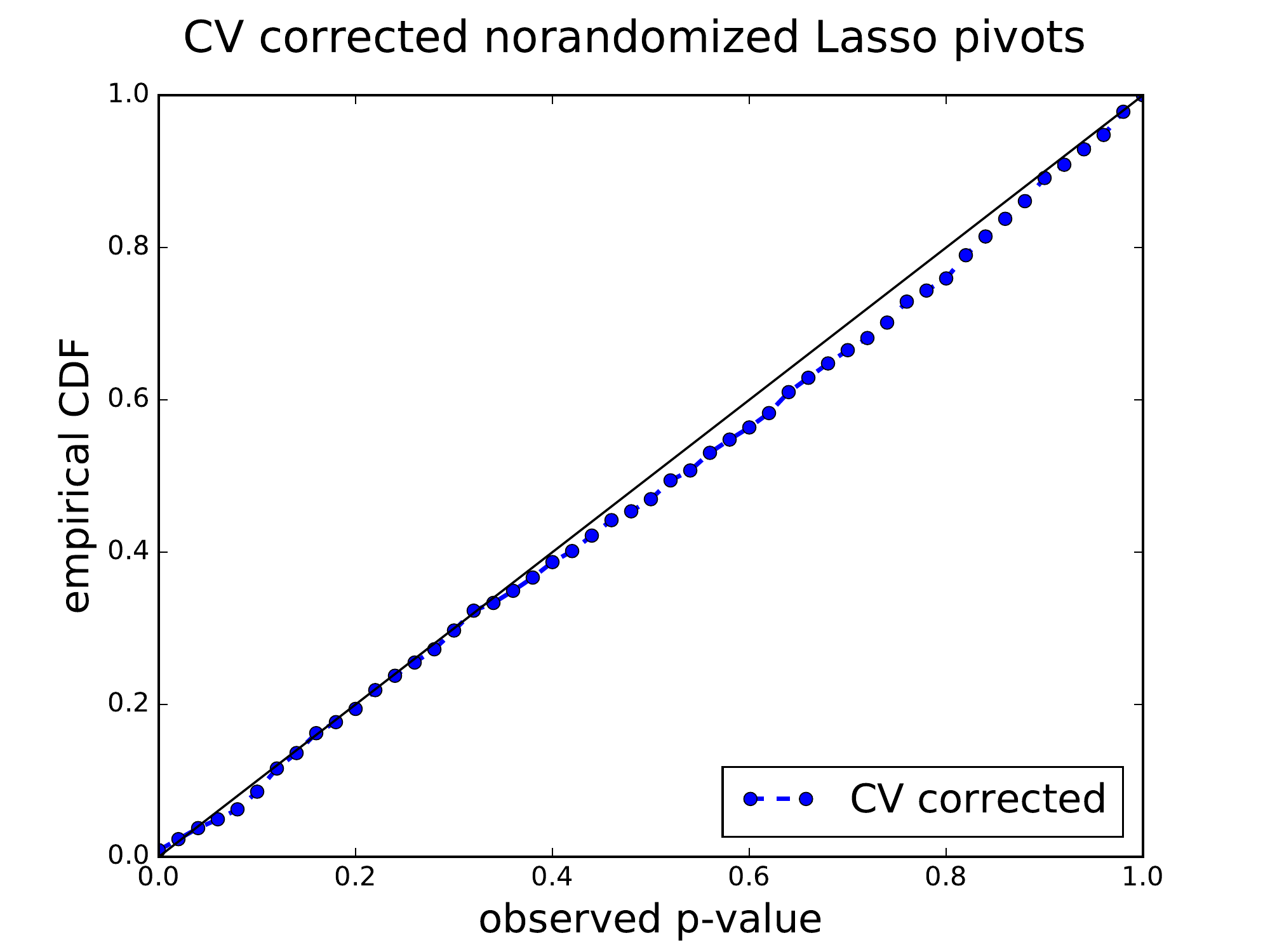}
    \caption{Selective $p$-values adjusted for cross-validation, using randomized CV curve. The average coverage is 90$\%$.} 
    \label{fig:cv_corrected_nonrandomized_lasso_pivots}
\end{figure}


\section{Inference after CV and randomized selection procedures} \label{sec:randomized:lasso:cv}

In this section, we apply an extension of our general framework developed in Section \ref{sec:gen:framework} to the problems of doing inference after randomized selection procedures. It differs from Section \ref{sec:gen:framework} since in the inference part, we marginalize over the added randomization in CV. As shown in \citealt{tian2016magic, selective_sampler}, inference after randomized model selection procedures has larger power than after their non-randomized counterparts. We further demonstrate this in Section \ref{sec:power} in the appendix.

We focus on the randomized Lasso with randomized cross-validation. We assume $\lambda_R^{cv}$ is computed based on randomized cross-validation, similar to Section \ref{sec:randomized:cv:curve}, with the details given below in Section \ref{sec:randomized:lasso:randomized:cv}. After choosing $\lambda_R^{cv}$, we solve a randomized Lasso objective as follows 
\begin{equation} \label{eq:lasso:objective:randomized}
	\hat{\beta}(X,y,\omega) = \textnormal{arg}\underset{\beta\in\mathbb{R}^p}{\min}\:\frac{1}{2}\left\|y-X\beta\right\|_2^2+\lambda_R^{cv}\|\beta\|_1+\frac{\nu}{2}\|\beta\|_2^2-\omega^\top \beta, 
\end{equation}
where $(X,y)\times\omega\sim \mathbb{F}_n\times \mathbb{F}_{\omega}$. $\omega$ is a randomization sample from the pre-specified distribution $\mathbb{F}_{\omega}$ with density $g_{\omega}$. $\nu$ is a small constant, ensuring the solution of the objective above exists. Similar to the non-randomized Lasso, the randomized objective above induces sparsity so we denote $\widehat{E}(X,y,\omega)=\left\{j\in\{1,\ldots,p\}:\hat{\beta}(X,y,\omega)_j\neq 0\right\}$, consisting of non-zero coefficients of the solution $\hat{\beta}(X,y,\lambda)$. After observing $\widehat{E}(X,y,\omega)=E$, we decide on the parameters of interest for inference.
In what follows, we explain how to provide inference for the parameters chosen based on observing the set of selected predictors $E$ after adjusting for both randomized Lasso and randomized cross-validation.

\subsection{Adjusting for randomized Lasso alone with fixed $\lambda$} 

 
After running the randomized Lasso, we provide inference based on looking at the set $E$ of non-zero coefficients of its solution $\hat{\beta}(X,y,\omega)$. In addition to adjusting for $E$, we condition on the signs $s_E$ of the active portion of the randomized Lasso solution as in non-randomized setting.
In order to have valid inference, we need to adjust for looking at these outcomes by conditioning on the observation that our data $(X,y)$ and randomization $\omega$ landed in the selection region given by
\begin{equation*}
	\left\{\left(X',y',\omega'\right)\in\mathbb{R}^{n\times p}\times\mathbb{R}^n\times\mathbb{R}^p: \textnormal{sign}(\hat{\beta}(X',y',\omega')_E)=s_E, \hat{\beta}(X',y',\omega')_{-E}=0\right\}.
\end{equation*}
In other words, we want to base our inference using the distribution of the data conditional on $(X,y,\omega)\sim\mathbb{F}_n\times\mathbb{F}_{\omega}$ landing in the selection region above. Getting this post-selection distribution of the data by directly sampling data and randomization from the set above is hard due to complicated joint constraints.

Following the trick of change of measure in \cite{tian2016magic, selective_sampler}, we do not sample directly $(X,y,\omega)$ from their conditional distribution. Instead, to get post-selection distribution of the data, we sample data and the so called, \textit{optimization variables} from a simpler selection event, depicted by optimization variables only. Since the randomized Lasso solution is a function of the vector $D$ as defined in \eqref{eq:D:vector} and randomization $\omega$, the following proposition expresses their conditional density.

\begin{proposition}[\cite{selective_sampler}] 
	Assume that pre-selection asymptotic density of $(D,\omega)$ is $\phi_{(\mu_D,\Sigma_D)}(D)\cdot g_{\omega}(\omega)$, where $\phi_{(\mu,\Sigma_D)}$ denotes the density of $\mathcal{N}(\mu_D,\Sigma_D)$. The conditional asymptotic density of $(D,\omega)$ given the Lasso selected model $E$ with the signs $s_E$ of the active coefficients can be expressed via a change of variables 
	$$
	\omega=MD+B\beta_E+\begin{pmatrix} \lambda s_E \\ u_{-E} \end{pmatrix},
	$$ 
	for $(\beta_E, u_{-E})\in\mathbb{R}^{|E|}\times\mathbb{R}^{p-|E|}$, $\textnormal{sign}(\beta_E)=s_E$, $\|u_{-E}\|_{\infty}\leq\lambda$, and $M$ and $B$ are the following matrices 
	\begin{equation*}
   	M = \begin{pmatrix} X_E^\top X_E & 0 \\ X_{-E}^\top X_E & I_{p-|E|} \end{pmatrix}, \;\;
   	B = \begin{pmatrix} X_E^\top X_E+\nu I_{|E|} \\ X_{-E}^\top X_E \end{pmatrix}.
	\end{equation*}
	The selective density of $(D,\beta_E, u_{-E})\in\mathbb{R}^p\times\mathbb{R}^{|E|}\times\mathbb{R}^{p-|E|}$ is then proportional to
	\begin{equation} \label{eq:selective:density:D}
		\phi_{(\mu_D,\Sigma_D)}(D)\cdot g_{\omega}\left(MD+B\beta_E+\begin{pmatrix} \lambda s_E \\ u_{-E} \end{pmatrix}\right)
	\end{equation}
	with the constraints $\textnormal{sign}(\beta_E)=s_E$ and $\|u_{-E}\|_{\infty}\leq\lambda$.
\end{proposition}
In other words, after sampling $(D,\beta_E, u_{-E})$ from the constrained density in \eqref{eq:selective:density:D}, 
$$
\left(D, MD+B\beta_E+\begin{pmatrix} \lambda  s_E \\ u_{-E}\end{pmatrix}\right)
$$ 
comes from the asymptotic density of $(D,\omega)$ conditional on $(E,s_E)$.

\begin{proof}
The proof follows easily from the KKT conditions of randomized Lasso are
\begin{equation*}
	\omega = MD+B\beta_E+\begin{pmatrix}
		\lambda s_E \\ u_{-E}
	\end{pmatrix}	
\end{equation*}
with the constraints $\textnormal{sign}(\beta_E)=s_E$ and $\|u_{-E}\|_{\infty}\leq\lambda$, $\beta_E= \hat{\beta}(X,y,\omega)_E$ corresponds to the active part of the solution; $u_{-E}$ corresponds to the inactive part of the sub-gradient penalty, $\partial(\lambda\|\beta\|_1)_{-E}$, evaluated at the solution $\hat{\beta}(X,y,\omega)$. 
\end{proof}


\begin{remark}
 Instead of sampling data vector $D$ and $\omega$, we sample $D$ together with so called optimization variables $(\beta_E, u_{-E})$ from \eqref{eq:selective:density:D}. Note that the optimization variables are restricted to a simple set, a product of orthans and cubes, and there are no constrains on the data vector $D$ in the selective density. We elaborate on the sampler used to sample from this density in Section \ref{subsec:sampling}.
\end{remark}

After observing set $E$, we choose the parameter of interest $\theta$ and the corresponding target statistic $T$ that is asymptotically normal with mean $\theta$ pre-selection. As described, to do proper adjustment, we need to base inference on the conditional distribution of $T$. Since the selective density above is in terms of $D$, we re-write it in terms of $T$. Assuming $(T,D)$ are jointly normal pre-selection we can decompose $D=\Sigma_{D,T}\Sigma_T^{-1}T+N_D$. By conditioning on $N_D$, we write the selective density in terms of $T$ and optimization variables $(\beta_E, u_{-E})$ as
\begin{equation} \label{eq:randomized:lasso:density}
	\phi_{(\theta,\Sigma_T)}(T)\cdot g_{\omega}\left(M\Sigma_{D,T}\Sigma_T^{-1}T+MN_D+B\beta_E+\begin{pmatrix} \lambda s_E \\ u_{-E}\end{pmatrix}\right),
\end{equation}
with the same constraints on $(\beta_E, u_{-E})$ as above.

\subsection{Adjusting for both randomized Lasso and randomized cross-validation} \label{sec:randomized:lasso:randomized:cv}

In addition to adjusting for the randomized Lasso selecting the set $E$ of predictors, we need to account for the fact that $\lambda_R^{cv}=\lambda_{r^*}$ has been chosen in a data dependent way, assuming we run randomized Lasso and randomized cross-validation on our data before doing inference. 

To do randomized cross-validation in this setting, we compute the curve composed of models' quality values, i.e. $Err$ (as in the cross-validated non-randomized Lasso setting in Section \ref{sec:nonrandomized:lasso:cv}).  Our framework applies to other ways of computing $Err$ as well, as long as the final randomized cross-validation curve satisfies some assumptions.
Given the vector $Err$, we compute the randomized cross-validation curve as
\begin{equation*}
	Err_R = Err+R_1+R_2=Err^{(1)}+R_2,
\end{equation*}
where $R_1\times R_2\sim\mathbb{F}_{R_1}\times\mathbb{\mathbb{F}}_{R_2}$ for $\mathbb{F}_{R_1}=\mathcal{N}(0,\sigma_{R_1}^2 I_L)$ with known $\sigma_{R_1}$ and $\mathbb{F}_{R_2}$ with density $g_{R_2}$ chosen in advance; $R_1$ and $R_2$ are independent of everything else. Note that we write the randomization as the sum of two randomization terms $R_1$ and $R_2$, where $R_1$ is normally distributed and the distribution of $R_2$ is pre-specified (taken to be log-concave for computational reasons) but not necessarily a normal distribution. We describe soon why we decompose the randomization in such a way. To account for randomized cross-validation, we need to condition on the event $Err_R\in\mathcal{E}_{r^*}$, that the index of the minimizer of $Err_R$ is $r^*$.
 
Now we have the selection event for the randomized Lasso written in terms of $D$ and $\omega$ and the selection event of randomized cross-validation written in terms of $Err^{(1)}$ and $R_2$ (we use the decomposition of $Err_R$ into a sum of $Err^{(1)}$ and $R_2$ purposefully). Given the parameter of interest $\theta=\theta(\mathbb{F}_n, E)$ and the corresponding target statistic $T=T((X,y),E)$, we rewrite the selective density in terms of $T$ and the randomizations $\omega$ and $R_2$ as follows. The proof of the following proposition consists of writing the randomization reconstruction for $\omega$ and $R_2$ so we omit it here.

\begin{proposition} 
Assuming $T$, $D$ and $Err^{(1)}$ are jointly asymptotically Gaussian pre-selection, i.e.~under $\mathbb{F}_n\times\mathbb{F}_{R_1}$ and fixed $E$
\begin{equation} \label{eq:randomized:cv:data:vector}
	\begin{pmatrix} T \\ D \\ Err^{(1)} \end{pmatrix} \overset{d}{\rightarrow} \mathcal{N}\left(\begin{pmatrix} \theta \\ \mu_D \\ \mu_{Err^{(1)}} \end{pmatrix}, 
	\begin{pmatrix} 
	\Sigma_T & \Sigma_{T,D} & \Sigma_{T,Err^{(1)}} \\
	\Sigma_{D,T} & \Sigma_D & \Sigma_{D, Err^{(1)}} \\ \Sigma_{Err^{(1)},T} & \Sigma_{Err^{(1)},D} & \Sigma_{Err^{(1)}} \end{pmatrix}\right)
\end{equation}
as $n\rightarrow\infty$. Denote $Err^{(1)}=\Sigma_{Err^{(1)}, T}\Sigma_T^{-1}T+N_{Err^{(1)}}$. 
The asymptotic post-selection density of $(T,\beta_E, u_{-E}, Err_R)$, where the conditioning is on $(E,s_E,r^*)$ and $(N_D,N_{Err^{(1)}})$, is proportional to
\begin{equation} \label{eq:density:cv:full}
\begin{aligned} 
	\phi_{(\theta,\Sigma_T)}(T) & \cdot g_{\omega}\left(M\Sigma_{D,T}\Sigma_T^{-1}T+MN_D+B\beta_E+\begin{pmatrix} \lambda s_E \\ u_{-E}\end{pmatrix}\right) \\
& \cdot\: g_{R_2}(Err_R - \Sigma_{Err^{(1)}, T}\Sigma_T^{-1}T-N_{Err^{(1)}}) 
\end{aligned}
\end{equation}
restricted to $(\beta_E, u_{-E}, Err_R)\in \mathbb{R}_{s_E}^{|E|}\times[-\lambda, \lambda]^{p-|E|}\times\mathcal{E}_{r^*}$. 
\end{proposition}

\begin{remark}
\begin{itemize}
\item[--] Notice that, in the sampling density \eqref{eq:density:cv:full}, $\beta_E$ and $u_{-E}$ correspond to the randomized Lasso constraint and $Err_R$ corresponds to the CV constraint. To do valid inference on $\theta$, it suffices to have the samples of $T$ from this density.

\item[--] Since the randomizations $\omega\sim\mathbb{F}_{\omega}$ and $R_2\sim\mathbb{F}_{R_2}$ are mutually independent and independent of everything else, $g_{\omega}$ and $g_{R_2}$ separate as written in \eqref{eq:density:cv:full}. Thus, we consider randomized Lasso and randomized cross-validation as two queries/views on the data, where one view corresponds to $g_{\omega}(\cdot)$ and the other corresponds to $g_{R_2}(\cdot)$ \citep{bootstrap_mv}. Considering both views, the optimization variables, defined as moving particles in the sampler other than the target, are $(\beta_E, u_{-E}, Err_R)$ in this case.
\end{itemize}
\end{remark}

\subsection{Computing the pivot via sampling}  \label{subsec:sampling}

We use Markov chain Monte Carlo (MCMC) methods to sample from the density in (\eqref{eq:density:cv:full}). Moving a particle via MCMC in high-dimensions is computationally infeasible, thus we will marginalize over the sub-gradient $u_{-E}$ explicitly in the density (\eqref{eq:density:cv:full}). This requires computing the volume of a cube under $\mathbb{F}_{\omega}$. Since $\omega\sim\mathbb{F}_{\omega}$ is chosen to consist of i.i.d.~components, this volume can be written as a product of individual components, and it is easy to compute analytically. The details are given in \cite{selective_sampler}. As a result, we sample only $(T,\beta_E, Err_R)$, which lie in dimension $\textnormal{dim}(\theta)+|E|+L$, with $L=|\Lambda|$ as the grid size for $\lambda$ in cross-validation. Usually, the target is $T=\bar{\beta}_E$ if the parameter of interest is $\theta=\beta_E^*=\left(\mathbb{E}_{\mathbb{F}_n}\left[X_E^\top X_E\right]\right)^{-1}\mathbb{E}_{\mathbb{F}_n}\left[X_E^\top y\right]$, thus $\textnormal{dim}(\theta)=|E|$. Since the size $|E|$ of the selected model is in general small, the sampling is  feasible using MCMC. We can further reduce the dimension of the sampler by conditioning on any of the optimization variables. Although conditioning reduces the power of our test, sometimes the difference is negligible.  

\begin{remark}
We can choose to either move $Err_R$ in the sampler or condition on it. In the case we do not condition on it, we sample $(T=\bar{\beta}_E,\beta_E,Err_R)$ from the density in \eqref{eq:density:cv:full} with the constraints on $\beta_E$ and $Err_R$. Conditioning on $Err_R$ means we fix it in the sampler at its observed value. As a result, we sample $(T=\bar{\beta}_E,\beta_E)$ from the density in \eqref{eq:density:cv:full} with the constraints on $\beta_E$ only. The latter sampling scheme requires only the projection of $\beta_E$ at each step. Furthermore, once we choose to condition on $Err_R$, we are ``allowed'' to look at all of its values. Thus, in this case we can choose $\lambda$ differently and not necessarily the minimizer of $Err_R$, e.g.~using one sigma rule up or down from the minimizer $\lambda_R^{cv}$ \citep{friedman2001elements}.
\end{remark}

Ideally, to construct confidence intervals, we will have to conduct tests at all different values of $\theta$, and decide whether we want to include them as part of the confidence interval. As this is computationally heavy, we adopt importance sampling to construct confidence intervals efficiently. Specifically, we do sampling only once under a reference parameter, and tilt the original samples to get them distributed under different $\theta$ values \citep{bootstrap_mv}.

As for the MCMC sampler, we use projected Langevin for which \citealt{bubeck2015sampling} gives theoretical guarantees. It allows us to sample from a log-concave density with constraints. At each step, the optimization variables are projected to their constraint set; since in our case these constraints are simple polyhedrons, sampling is computationally fast. Projecting $\beta_E$, $u_{-E}$ onto $\mathbb{R}_{s_E}^{|E|}$ and $[-\lambda,\lambda]^{p-|E|}$ respectively is simple and the details of the projection of $Err_R$ onto $\mathcal{E}_{r^*}$ are given in Section \ref{sec:sampling} in the appendix.


\subsection{Randomized selective pivot}  \label{sec:randomized:pivot}

We now define the randomized pivot and prove it is valid after selection. Going back to the general setting of Section \ref{sec:gen:framework}, recall that $\widetilde{D}=\widetilde{D}(\mathbb{F}_n,M)$ is a general data vector, e.g.~$(D, Err^{(1)})$ in the example above, and $\tilde{\omega}$ contains all the randomization used in the procedure and $\mathbb{F}_{\tilde{\omega}}$ is the joint distribution of $\tilde{\omega}$. In the randomized Lasso example above with randomized cross-validation, $\mathbb{F}_{\tilde{\omega}}$ becomes $\mathbb{F}_{\omega}\times\mathbb{F}_{R_1}\times\mathbb{F}_{R_2}$. The randomization $\tilde{\omega}$ is independent of the data vector $\widetilde{D}$. We assume the selection event is affine in terms of $(\widetilde{D},\tilde{\omega})$, i.e.~can be represented as
\begin{equation} \label{eq:affine:randomized}
	\widetilde{A}\begin{pmatrix} \widetilde{D} \\ \tilde{\omega} \end{pmatrix} \leq \tilde{a}_n	
\end{equation}
for a sequence $\tilde{a}_n \rightarrow \tilde{a}$ as $n\rightarrow\infty$. In the example above, this constraint is written in terms of the optimization variables $O=O(\widetilde{D},\tilde{\omega})$ constrained to lie in a set $\mathcal{O}$ after selection. 

Given the parameter of interest $\theta$ and the target statistic $T$, we assume \eqref{eq:D:T:gaussian} holds, allowing us to condition on $N_{\widetilde{D}}=\widetilde{D}-\Sigma_{\widetilde{D},T}\Sigma_T^{-1}T$, the statistic corresponding to the nuisance parameters.
We define \textit{the randomized selective pivot} $\mathcal{P}^R$ as 
\begin{equation*}
\begin{aligned}
	& \mathcal{P}^R\left((T,\widetilde{D});\widetilde{A},\tilde{a}_n\right) \\
	& = \mathbb{P}_{\Phi\times\mathbb{F}_{\tilde{\omega}}}\left\{\left\|Z_T-\theta\right\|_2\leq \left\|T-\theta\right\|_2\:\Big|\:\widetilde{A}\begin{pmatrix} Z_{\widetilde{D}} \\ \tilde{\omega}\end{pmatrix} \leq \tilde{a}_n, Z_{\widetilde{D}}-\Sigma_{\widetilde{D},T}\Sigma_T^{-1}Z_T=N_{\widetilde{D}}\right\} \\
	&= \mathbb{P}_{\Phi\times\mathbb{F}_{\tilde{\omega}}}\left\{\left\|Z_T-\theta\right\|_2\leq \left\|T-\theta\right\|_2\:\Big|\:\widetilde{A}\begin{pmatrix} \Sigma_{\widetilde{D},T}\Sigma_T^{-1}Z_T+N_{\widetilde{D}} \\ \tilde{\omega}\end{pmatrix} \leq \tilde{a}_n, Z_{\widetilde{D}}-\Sigma_{\widetilde{D},T}\Sigma_T^{-1}Z_T=N_{\widetilde{D}}\right\} \\
	&=\mathbb{P}_{\Phi\times\mathbb{F}_{\tilde{\omega}}}\left\{\left\|Z_T-\theta\right\|_2\leq \left\|T-\theta\right\|_2\:\Big|\:O(\widetilde{D},\tilde{\omega})\in\mathcal{O}, Z_{\widetilde{D}}-\Sigma_{\widetilde{D},T}\Sigma_T^{-1}Z_T=N_{\widetilde{D}}\right\},
\end{aligned}
\end{equation*}
where the probability is under $\left((Z_T,Z_{\widetilde{D}}), \tilde{\omega}\right)\sim\Phi\times\mathbb{F}_{\tilde{\omega}}$. Note that in the definition of the pivot above we marginalize over randomization $\widetilde{\omega}\sim\mathbb{F}_{\tilde{\omega}}$.
Recall $\Phi$ denotes the asymptotic Gaussian distribution of the data $(T,\widetilde{D})$ pre-selection.

We present a theorem stating the randomized pivot above is asymptotically $\textnormal{Unif}(0,1)$. Consequently, we need the dimensions of $\widetilde{D}$, $\tilde{\omega}$ and $T$ to be fixed, which in the randomized Lasso example translates to fixed $p$. Since the proof is analogous to the proof of Theorem \ref{thm:nonrandomized:pivot} we omit the proof of the following theorem.

\begin{theorem}[Valid randomized selective pivot] \label{thm:randomized:pivot}
	Assuming \eqref{eq:affine:randomized} and \eqref{eq:D:T:gaussian} hold, we have under $(T,\widetilde{D})\sim\mathbb{F}_n^*$
	\begin{equation*}
		\mathcal{P}^R\left((T,\widetilde{D});\widetilde{A},\tilde{a}_n\right)\overset{d}{\rightarrow}\textnormal{Unif}(0,1) 
	\end{equation*}
	as $n\rightarrow\infty$, or equivalently, under $((T,\widetilde{D}),\tilde{\omega})\sim\mathbb{F}_n\times\mathbb{F}_{\tilde{\omega}}$
	\begin{equation*}
		\mathcal{P}^R\left((T,\widetilde{D});\widetilde{A}, \tilde{a}_n\right)\Big|\:O(\widetilde{D}, \tilde{\omega})\in\mathcal{O} \overset{d}{\rightarrow}\textnormal{Unif}(0,1)
	\end{equation*}
	as $n\rightarrow\infty$.
\end{theorem} 

\begin{remarks}
\begin{itemize}[leftmargin=*]
\item[--] \citealt{tian2015selective, bootstrap_mv} also proved that $\mathcal{P}^R$ is asymptotically $\textnormal{Unif}(0,1)$ in low-dimensional and non-parametric setting with a different set of assumptions on the selection event.	
\item[--] To make the randomized selective pivot applicable in practice we need the post-selection consistency of the estimates of the covariance matrices. We refer the reader to the results of \cite[Lemma 3]{tian2015selective} and \cite[Lemma 18]{bootstrap_mv} for the results similar to Lemma \ref{lemma:consistency} in the randomized setting, where we marginalize over the added randomization. These results make the covariance estimates both pre-selection and post-selection consistent.
\end{itemize}
\end{remarks}


\subsection{Simulation examples}

We empirically demonstrate the performance of using selective sampler to carry out valid inference, for randomized Lasso with randomized cross-validation. We report the results for $\ell_2$ and logistic loss together with Gaussian randomization in particular. But again, any convex loss and any log-concave randomization falls into our framework. 

\noindent\textbf{Simulation setup}: The entries of $X$ are independent standard normal random variables with the columns of $X$ normalized to have empirical variance 1. In the case of $\ell_2$ loss, the response $y$ is generated from $y\sim\mathcal{N}(0, I_n)$, i.e.~a null signal with true sparsity $s=0$, independently of $X$. In the case of logistic loss, $y_i\overset{i.i.d}\sim\textnormal{Bernoulli}(1/2)$, $i=1,\ldots, n$. We take $n=600, p=100$. After we select the model $E$ by running randomized Lasso with $\lambda$ as the minimizer of $Err_R$, we compute the coverage by checking how many of the constructed intervals cover zero. Since the true sparsity is zero, checking the coverage is easy since $\beta_E^*=0$ in this setting. In general case, when the true parameter is with non-empty support, we would need to first check whether the selected model is a superset of the true model before checking the coverage.

The selective pivots have been constructed by sampling $(T=\bar{\beta}_E,\beta_E, Err_R)$ from the density in \eqref{eq:density:cv:full} with $\ell_2$ loss in Figure \ref{fig:random:cv:gaussian:loss} and with logistic loss in Figure \ref{fig:randomi:cv:logistic:loss}. After sampling we discard the samples of $(\beta_E,Err_R)$ and keep only the samples of $T$ to do inference. These experiments are repeated 100 times. 
The figures show the uniformity of selective $p$-values (in blue). We also show (in red) the empirical distribution of the naive $p$-values constructed based on the asymptotic normality of $T$ pre-selection. All the covariances are estimated using pairs bootstrap.

For more simulations, see Section \ref{sec:power} in the appendix. There we apply the traditional methods controlling multiple-testing errors, e.g.~Benjamini-Hochberg, to the selective $p$-values. 
We show good empirical control of false discovery rate (FDR) although we do not have theoretical guarantees of FDR control. To illustrate high statistical power of the selective $p$-values, we compare our method with knockoffs of \citealt{original_knockoffs} (in simulations favorable to knockoffs), which is known to have excellent FDR control and high power.

\begin{figure}[h!]
\centering
    \includegraphics[width=0.60\textwidth]{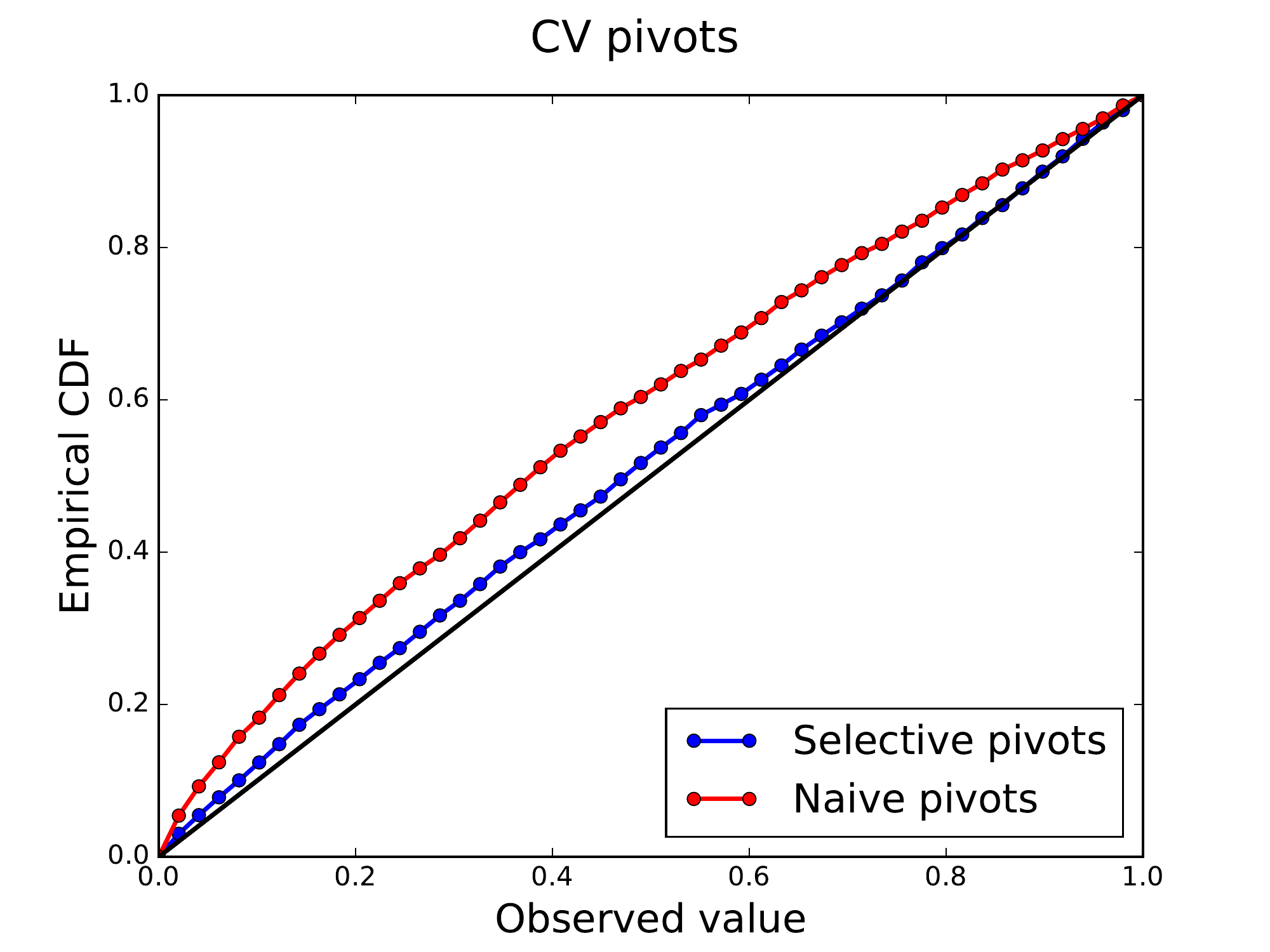}
    \caption{The coverage and the average length of selective intervals are 88$\%$ and 4.1; and for the naive ones are 82$\%$ and 3.25, respectively.}
    \label{fig:random:cv:gaussian:loss}
\end{figure}
  
\begin{figure}[h!]
  \centering  
    \includegraphics[width=0.60\textwidth]{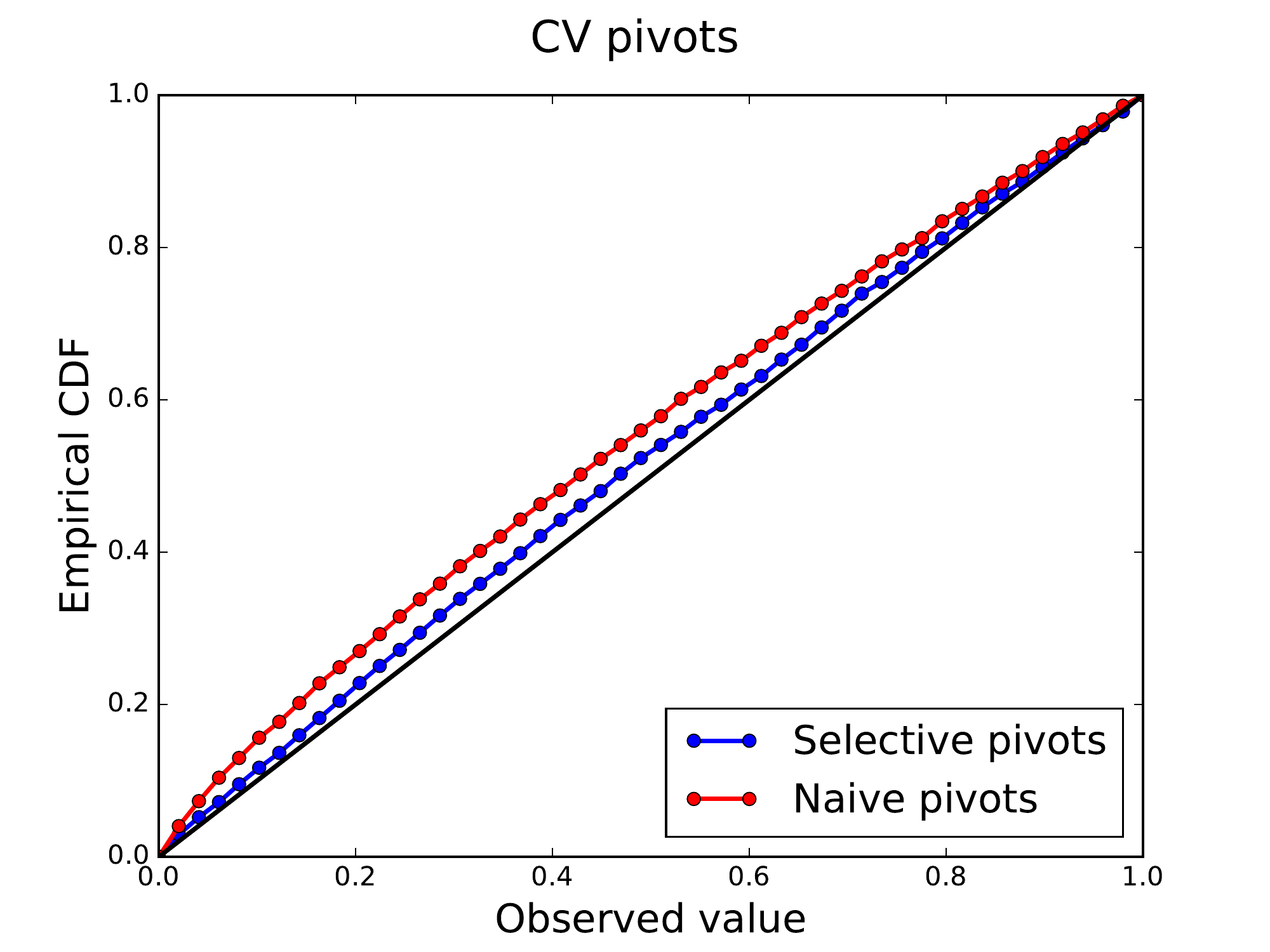}
    \caption{The selective pivots have been constructed by sampling $(T,\beta_E)$ from the density in \eqref{eq:density:cv:full} . The coverage and the average length of selective intervals are 88$\%$ and 8.8,; and for the naive ones  are 85$\%$ and 7.6, respectively.}
    \label{fig:randomi:cv:logistic:loss}
\end{figure}


\section{Inference after choosing a model based on AIC criteria} \label{sec:AIC}

To further demonstrate the applicability of our approach, we present a way of doing inference for the selected coefficients after choosing a model based on AIC criteria. 

\subsection{Inference after choosing a model from a given set of possible models}

Similar to cross-validation approach, we randomize the vector composed of AIC criteria values for each of the given models, in order to make it asymptotically jointly Gaussian. Furthermore, by requiring it to be jointly asymptotically Gaussian with the target statistic, we can  decompose randomized AIC vector with respect to the target statistics. For simplicity we take the set of models to be fixed in advance, i.e.~before looking at the data and the only selection adjustment is done by looking at one or more minimizers of the randomized AIC criteria vector.

Given the data $(X,y)\in\mathbb{R}^{n\times p}\times\mathbb{R}^n$ and a set of models $\mathfrak{E}=\{E_1,\ldots,E_L\}$, $E_l\subset\{1,\ldots,p\}$ for each $l=1,\ldots, L$, assuming Gaussian likelihood, we evaluate the AIC criterion for each of them as
\begin{equation*}
	Err_l = \frac{1}{n}\left\|y-X_{E_l}\bar{\beta}_{E_l}\right\|_2^2+a(|E_l|) \;\;\textnormal{ or }\;\;
	Err_l = \frac{1}{n}\left\|y-X_{E_l}\bar{\beta}_{E_l}\right\|_2^2\cdot a(|E_l|),
\end{equation*}
where $\bar{\beta}_{E_l}$ is the OLS estimator with the response $y$ and design matrix $X_{E_l}$ (keeping only predictors from $E_l$) and $a(|E_l|)$ is a function penalizing the size of the model $E_l$. Note that for Gaussian likelihood, AIC criterion can be in one of the two forms above, depending on whether the standard error of the residuals is known or not. Since $a(|E_l|)$, $l=1,\ldots, L$, is a sequence of constants, we omit these terms when proving asymptotic normality of the randomized AIC criteria vector (Section \ref{sec:AIC:proof}). Let us denote the vector of all AIC criteria as 
\begin{equation*}
	Err = \left(Err_1,\ldots, Err_L\right),
\end{equation*}
and its randomized version as
\begin{equation} \label{eq:randomized:AIC}
	Err_R = Err+R_1+R_2 = Err^{(1)}+R_2,
\end{equation}
where $R_1\times R_2\sim\mathcal{N}(0,\sigma_{R_1}^2)\times\mathbb{F}_{R_2}$, chosen as in Section \ref{sec:randomized:lasso:cv}. By looking at the indices $\mathcal{L}^*=\left\{l_1^*,\ldots, l_K^*\right\}$ of $K$ smallest values of $Err_R$, with $K$ specified in advance, we choose $K$ models from $\mathfrak{E}$ compromising the set $\mathfrak{E}^*=\{E_{l_1^*},\ldots, E_{l_K^*}\}\subset\mathfrak{E}$. Based on $\mathfrak{E}^*$, we choose the parameter of interest $\theta$ and the target test statistic $T$, asymptotically with mean $\theta$ and jointly asymptotically normal with $Err^{(1)}$.

In order to do valid inference post-selection for $\theta$, we want to condition on the fact that the $K$ smallest values of vector $Err_R$ are achieved at $\mathcal{L}^*$. This corresponds to an affine constraint on $Err_R$, hence we write it as 
\begin{equation*}
	B_{\mathcal{L}^*} \cdot Err_R \leq 0,
\end{equation*}
where $B_{\mathcal{L}^*}$ is a matrix that only depends on the set $\mathcal{L}^*$.

Given the parameter of interest $\theta$ chosen after looking at $\mathcal{L}^*$, assume we have the target statistic $T$ that is asymptotically normal pre-selection, treating $\mathcal{L}^*$ as fixed and not chosen based on data. In order to do valid inference for $\theta$ using $T$, we condition the distribution of $T$ on observing $\mathcal{L}^*$. The following proposition characterizes the post-selection distribution of $T$ used for inference.

\begin{proposition}
Assume $T$ and $Err^{(1)}$ are jointly asymptotically normal pre-selection, i.e.
\begin{equation*}
	\begin{pmatrix} T\\ Err^{(1)} \end{pmatrix}
	\overset{d}{\rightarrow} \mathcal{N} \left(\begin{pmatrix} \theta \\ \mu_{Err^{(1)}} \end{pmatrix}, \begin{pmatrix} \Sigma_T & \Sigma_{T,Err^{(1)}} \\ \Sigma_{Err^{(1)}, T} & \Sigma_{Err^{(1)}} \end{pmatrix} \right)
\end{equation*}
as $n\rightarrow\infty$. Denote $N_{Err^{(1)}} = Err^{(1)}-\Sigma_{Err^{(1)},T}\Sigma_T^{-1}T$. 
The asymptotic selective density on $(T,Err_R)$ conditional on $\mathcal{L}^*$ is proportional to
\begin{equation} \label{eq:AIC:selective:density}
	\phi_{(\theta,\Sigma_T)}(T)\cdot g_{R_2}\left(Err_R-\Sigma_{Err^{(1)},T}\Sigma_T^{-1}T-N_{Err^{(1)}}\right)
\end{equation}
with the restriction $Err_R\in\mathcal{E}_{\mathcal{L}^*}=\left\{Err_R'\in\mathbb{R}^L:B_{\mathcal{L}^*} \cdot Err_R'\leq 0\right\}$.
\end{proposition}

Recall that we assume $\mathfrak{E}$ and $K$ are fixed in advance, i.e.~before looking at the data. It is, however, not hard to adjust our inference when we choose these parameters in data-dependent ways. One example is presented in the next section.


\subsection{Forward-stepwise with data-dependent number of steps} \label{sec:FS:main}

In most of the selective inference literature the number of steps $L$ in forward-stepwise (FS) algorithm has been kept fixed. The selective inference is done to adjust only for the selected predictors and not for the number of steps, except in \cite{sequential_post_selection}, where they developed a polyhedral description for choosing $L$ in a data-dependent way for non-randomized FS. After adding randomness in the objective of FS (randomized FS algorithm), selective inference adjustments for fixed $L$ are presented in \cite{selective_sampler} and its bootstrap version in \cite{bootstrap_mv}. We explain how the general method we introduced can be applied to do inference in cases where the number of steps $L$ is chosen in a data-dependent way.

For $l=1,2,\ldots,$ we solve a randomized or non-randomized FS algorithm on the data $(X,y)$ to select a set of predictors, denoted as $E_l$. The FS objective deciding which variables to include at a given step, might have a different loss function from the objective which decides when to stop adding variables. In order to decide whether to include the set $E_l$ chosen at step $l$ or stop at all variables chosen before the step $l$, we compute the resulting FS-criterion, $Err_l\in\mathbb{R}$, of the selected model chosen up to step $l$ including $l$. $Err_l$ is usually the AIC criterion of the model $\cup_{l'=1}^lE_{l'}$, the union of the selected predictors at steps $l'=1,\ldots, l$.

Denote the randomized version of $Err_l$, $1\leq l\leq \textnormal{min}\{n,p\}$ as $Err_{R,l}$, computed as in \eqref{eq:randomized:AIC}. We choose $L$ to be the first index $l$ for which $Err_{R,l}$, does not change much compared to $Err_{R,l-1}$, i.e.~it satisfies $\left\{Err_{R,l}\leq \eta\cdot Err_{R,l-1}\right\}$ for a pre-specified threshold $\eta$. Precisely,
\begin{equation*}
	L=\min\left\{2\leq l\leq\min\{n,p\} \;\textnormal{ s.t.}\; Err_{R,l}\leq \eta \cdot Err_{R,l-1} \right\}.
\end{equation*}
Usually $\eta$ is chosen to be close to 1 and smaller than 1, implying the randomized FS criterion, $Err_{R,l}$, does not decrease much at step $l$ comparing to the randomized FS criterion at step $l-1$, $Err_{R,l-1}$. Other variants of choosing the stopping time, e.g.~running the FS algorithm for all $\textnormal{min}\{p,n\}$ steps and then choose the final model based on the minimum of the randomized FS-criterion evaluated at each of the models, is also doable via the same framework.

Having described the model selection procedure, we turn to doing inference having looked at the outcomes.
The adjustment is done similarly to the cross-validation example. If we run the non-randomized FS at each step to add a predictor, we describe the selection event by appending two sets of constraints. The first set represents the affine constraints coming from $Err_R=\left(Err_{R,1},\ldots, Err_{R,L}\right)$, accounting for choosing $L$ in a data-dependent way. The second set of constraints accounts for choosing set $E_l$ at each step $l=1,\ldots, L$; \cite{sequential_post_selection} describes this selection event that treats $L$  as a constant in detail. In order to have valid inference, we have to account for both of the events.

On the other hand, if we run randomized FS algorithm at each of the $L$ steps, we use the selective sampler to get the target samples for its post-selection distribution. In this scenario, we take the randomization added to $Err=\left(Err_1,\ldots, Err_L\right)\in\mathbb{R}^L$ vector to be independent of the randomizations added in the FS objective at each step (when choosing the variables). Having these randomizations to be independent allows for a simple selective density since the randomizations densities separate. In other words, we have $L+1$ views/queries on the data: $L$ of which come from running the total of $L$ FS algorithms to pick the variables, and one view coming from the constraints on the vector $Err_R$, accounting for choosing the stopping time. Section \ref{sec:FS:details} in the appendix describes this scenario in detail including the full description for the selective density for the target statistic after performing FS with the data-dependent number of steps.


\subsection{Proving randomized AIC criteria vector is asymptotically jointly normal} \label{sec:AIC:proof}

We prove $Err_R$ defined in \eqref{eq:randomized:AIC} is an asymptotically jointly normal vector under the assumptions to be stated. To simplify the notation we assume the set $\mathfrak{E}$ consists of two models $E_1$ and $E_2$.
The data $(X,y)\in\mathbb{R}^{n\times p}\times\mathbb{R}^n$ consists of $n$ i.i.d.~samples $(x_i,y_i)\sim\mathbb{F}$.
Given a selected set $E_i$, $i=1,2$, denote with $\bar{\beta}_{E_i}=\left(X^\top_{E_i}X_{E_i}\right)^{-1}X_{E_i}^\top y$, the OLS estimator $y\sim X_{E_i}$. Then we have
\begin{equation}
	\left\|y-X_{E_i}\bar{\beta}_{E_i}\right\|_2^2 =\left\|\left(I-P_{X_{E_i}}\right)y\right\|_2^2 = \left\|\left(I-P_{X_{E_i}}\right)\epsilon_{E_i}\right\|_2^2=\left\|\epsilon_{E_i}\right\|_2^2-\left\|P_{X_{E_i}}\epsilon_{E_i}\right\|_2^2,
\end{equation}
where $P_{X_{E_i}}=X_{E_i}\left(X^{\top}_{E_i} X_{E_i}\right)^{-1}X^{\top}_{E_i}$, $\epsilon_{E_i} = y-X_{E_i}\beta_{E_i}^*$ and $\beta_{E_i}^*=\left(\mathbb{E}\left[X_{E_i}^\top X_{E_i}\right]\right)^{-1}\mathbb{E}\left[X_{E_i}^\top y\right]$ are the population OLS parameters corresponding to $E_i$. Notice that the coordinates of $\epsilon_{E_i}=(\epsilon_{E_i,1},\ldots, \epsilon_{E_i,n})$ are i.i.d. Let us denote the following variances as $\textnormal{Var}\left(\epsilon_{E_i,1}^2\right)=\sigma_{E_i}^2$, $i=1,2$, and their cross covariance as $\textnormal{Cov}\left(\epsilon_{E_1,1}^2, \epsilon_{E_2,1}^2\right)=\sigma_{E_1,E_2}^2$. We assume all of them are finite.

Let us first provide an intuitive explanation why $\left\|y-X_{E_1}\bar{\beta}_{E_1}\right\|_2^2$ and $\left\|y-X_{E_2}\bar{\beta}_{E_2}\right\|_2^2$ might not be asymptotically jointly Gaussian without added randomization.
The difference of these two quantities is 
\begin{equation*}
\begin{aligned}
	&\left\|y-X_{E_1}\bar{\beta}_{E_1}\right\|_2^2-\left\|y-X_{E_2}\bar{\beta}_{E_2}\right\|_2^2 = \left\|\left(I-P_{X_{E_1}}\right)y\right\|_2^2-\left\|\left(I-P_{X_{E_2}}\right)y\right\|_2^2 \\
	&=\left\|\epsilon_{E_1}\right\|_2^2-\left\|P_{X_{E_1}}\epsilon_{E_1}\right\|_2^2-\left\|\epsilon_{E_2}\right\|_2^2+\left\|P_{X_{E_2}}\epsilon_{E_2}\right\|_2^2.
\end{aligned}
\end{equation*}
In general, the quantity above is not Gaussian. A simple illustration is as follows. Take the model for $\mathbb{F}^n$ to be a linear model $y=X\beta+\epsilon$ with $\epsilon|X\sim\mathcal{N}(0,I_n)$. Assuming $\textnormal{supp}(\beta)$ to be a subset of both $E_1$ and $E_2$, we have $\epsilon_{E_1}=\epsilon_{E_2}=\epsilon$.
Then the above quantity becomes $\epsilon^{\top}\left(P_{X_{E_2}}-P_{X_{E_1}}\right)\epsilon$. We know that $\left(P_{X_{E_2}}-P_{X_{E_1}}\right)$ is symmetric and $\text{Rank}\left(P_{X_{E_2}}-P_{X_{E_1}}\right)\leq \text{Rank}\left(P_{X_{E_1}}\right)+\text{Rank}\left(P_{X_{E_2}}\right)\leq |E_1|+|E_2|$. Assuming $|E_1|+|E_2|$ does not grow with $n$, the difference between testing errors is distributed as $\chi^2$ with a fixed degree of freedom instead of Gaussian. Therefore, we add randomization to both quantities to make the joint asymptotic normality possible. 

Before stating the main theorem of this section, let us state the assumptions needed.

\begin{itemize}
\item \textbf{Consistency assumption:} We assume $\sqrt{n}\left\|\bar{\beta}_{E_i}-\beta_{E_i}^*\right\|_2^2\overset{P}{\rightarrow}0 $ as $n\rightarrow\infty$.
\item \textbf{Moment assumption:} We assume $\sigma_{E_1}$, $\sigma_{E_2}$, $\sigma_{E_1,E_2}$, $\mathbb{E}\left[\left\|x_{1, E_i}\right\|_2^2\right]$, $i=1,2$, are all finite. $x_{1,E_i}\in\mathbb{R}^{|E_i|}$ denotes a sub-vector of $x_1\in\mathbb{R}^p$ corresponding to the coordinates in $E_i$, $i=1,2$.
\end{itemize}

\begin{theorem}[Randomized AIC curve is asymptotically jointly normal] \label{thm:AIC}
Assume the above consistency and moment assumptions hold, we have
\begin{equation}
	\frac{1}{\sqrt{n}}\left(\begin{pmatrix} \left\|y-X_{E_1}\bar{\beta}_{E_1}\right\|_2^2 \\ \left\|y-X_{E_2}\bar{\beta}_{E_2}\right\|_2^2\end{pmatrix}-\begin{pmatrix} \mathbb{E}\left[\left\|\epsilon_{E_1}\right\|_2^2\right]\\ \mathbb{E}\left[\left\|\epsilon_{E_2}\right\|_2^2\right]\end{pmatrix}\right)+\begin{pmatrix}
		R_1 \\ R_2
	\end{pmatrix} 
	\overset{d}{\rightarrow} \mathcal{N}\left(\begin{pmatrix} 0 \\ 0\end{pmatrix},\begin{pmatrix} \sigma_{E_1}^2+\tau_{R_1}^2 & \sigma_{E_1,E_2}^2 \\ \sigma_{E_1,E_2}^2 & \sigma_{E_2}^2+\tau_{R_2}^2\end{pmatrix}\right)
\end{equation}
as $n\rightarrow\infty$, where $R_1\sim\mathcal{N}\left(0,\tau_{R_1}^2\right)$ and $R_2\sim\mathcal{N}\left(0,\tau_{R_2}^2\right)$ are independent Gaussian random variables and independent of the data $(X,y)$; $\tau_{R_1}$ and $\tau_{R_2}$ are constants.
\end{theorem}

The proof of the theorem is given in Section \ref{sec:proofs} in the appendix.

\begin{remark}
	Note that when the population residuals $\epsilon_{E_1}$ and $\epsilon_{E_2}$, corresponding to the selected sets of predictors $E_1$ and $E_2$ respectively, are not that same, the difference $\left\|y-X_{E_1}\bar{\beta}_{E_1}\right\|_2^2$ $-\left\|y-X_{E_2}\bar{\beta}_{E_2}\right\|_2^2$ is not independent of the data. Hence in general the selection coming from looking at the minimum among the randomized version of the vector $\bigl(\left\|y-X_{E_1}\bar{\beta}_{E_1}\right\|_2^2,$ $ \left\|y-X_{E_2}\bar{\beta}_{E_2}\right\|_2^2\bigr)$ cannot be ignored. 
\end{remark}


\section{Marginal LOCO parameter} \label{sec:loco}

So far, we have discussed adding randomization to the criteria vector, whether it was cross-validation curve or AIC criteria curve, to make it asymptotically jointly Gaussian. Recall that for valid post-selection inference we need the target statistic to be asymptotically Gaussian as well pre-selection. This requirement was true in all of the examples so far since we have taken the target statistic to be the OLS estimator. In this section, however, the target parameter of interest is no longer the population regression coefficient $\beta_E^*$ but the marginal LOCO parameter to be defined in this section. Consequently, the target statistic also changes and in this case we need to add randomization to it to make it asymptotically jointly Gaussian since it might not satisfy that requirement otherwise.

For a given fixed set of variables $E\subset\{1,\ldots, p\}$, recall that the population regression parameter $\beta_E^*$ is defined as
$\beta_E^* = \left(\mathbb{E}_{(x,y)\sim\mathbb{F}}\left[x_E x_E^\top\right]\right)^{-1}\mathbb{E}_{(x,y)\sim\mathbb{F}}\left[yx_E\right]$, where the expectations are under a single data pair $(x, y)\sim\mathbb{F}$, $x\in\mathbb{R}^p$ and $y\in\mathbb{R}$. We define the \textit{marginal LOCO} parameter for the $j$-th predictor as
\begin{equation*}
	\gamma_j(\mathbb{F},E)=\mathbb{E}_{(x, y)\sim\mathbb{F}}\left[\left(y-x_{E\setminus j}^\top\beta_{E\setminus j}^*\right)^2-\left(y-x_E^\top\beta_E^*\right)^2\right], \;\; j\in E,
\end{equation*}
where the expectation is over a data pair $(x,y)\sim\mathbb{F}$. $x_{E\setminus j}$ is a vector computed by leaving the $j$-th covariate out from $x_E$ and $\beta_{E\setminus j}^*= \left(\mathbb{E}_{(x,y)\sim\mathbb{F}}\left[x_{E\setminus j} x_{E\setminus j}^\top\right]\right)^{-1}\mathbb{E}_{(x,y)\sim\mathbb{F}}\left[yx_{E\setminus j}\right]$ is the population regression parameter from only using the covariates in $E\setminus j$. The marginal LOCO parameter defined above measures the importance of a single predictor among the selected ones. By computing the difference in the second moment of the true residuals $y-x^{\top}_E\beta_E^*$ including the $j$-th predictor and the true residuals $y-x_{E\setminus j}^\top\beta_{E\setminus j}^*$, $\gamma_j(\mathbb{F}, E)$ measures the influence of a particular predictor in reducing $\ell_2$ loss.

\begin{remark}
The construction of the marginal LOCO parameter is inspired by the Leave out covariate (LOCO) quantity of \cite{rinaldo2016bootstrapping}.
Let $(X_i^{(1)},y^{(1)}_i)\overset{i.i.d.}{\sim}\mathbb{F}$, $i=1,\ldots, n_1$, and $(X_i^{(2)}, y^{(2)}_i)\overset{i.i.d.}{\sim}\mathbb{F}$, $i=1,\ldots, n_2$, be two independent samples denoted as $(X^{(1)},y^{(1)})\in\mathbb{R}^{n_1\times p}\times\mathbb{R}^{n_1}$ and $(X^{(2)}, y^{(2)})\in\mathbb{R}^{n_2\times p}\times\mathbb{R}^{n_2}$ of sizes $n_1$ and $n_2$, respectively.
For any predictor $j=1,\ldots,p$, the \textit{conditional LOCO} of \cite{rinaldo2016bootstrapping} is defined as
\begin{equation} \label{eq:conditional:LOCO}
	\gamma_j\left(\mathbb{F},\hat{\beta}^{(1)}, \hat{\beta}^{(1)}_{-j}\right) = \mathbb{E}_{(x,y)\sim\mathbb{F}}\left[\left|y-x^\top\hat{\beta}^{(1)}\right|-\left|y-x^\top\hat{\beta}^{(1)}_{-j}\right|\:\Big|\:\hat{\beta}^{(1)}, \hat{\beta}^{(1)}_{-j}\right],
\end{equation}
where the estimator $\hat{\beta}^{(1)}=\hat{\beta}^{(1)}(X^{(1)}, y^{(1)})$ is computed using $(X^{(1)}, y^{(1)})$ and the estimator $\hat{\beta}^{(1)}_{-j}=\hat{\beta}^{(1)}_{-j}(X^{(1)}_{-j}, y^{(1)})$ is computed on $(X^{(1)}_{-j}, y^{(1)})$ with $X_{-j}^{(1)}$ denoting the data $X^{(1)}$ without the $j$-th covariate. For simplicity, $\hat{\beta}^{(1)}_{-j}$ has an appended zero at the $j$-th coordinate. The expectation in  \eqref{eq:conditional:LOCO} is over one data point $(x,y)\sim\mathbb{F}$. 
The conditional LOCO  measures the influence of a particular predictor, $X_j$, on a prediction error of an estimator, $\hat{\beta}^{(1)}$. However, this measure is conditional on observing a particular training data, hence does not take into account the variance of the estimator and the conditional LOCO  of \cite{rinaldo2016bootstrapping} remains a random variable conditional on selection. 
\end{remark}

We build two test statistics that can be used for inference on $\gamma_j(\mathbb{F},E)$ given that the set $E$ is selected based on the data using a model selection algorithm. Given the training data $(X^{(1)},y^{(1)})\sim\mathbb{F}^{n_1}$ and the test data $(X^{(2)},y^{(2)})\sim\mathbb{F}^{n_2}$, we use a model selection algorithm applied to the training data only to get an active set $E$. The first test statistics, denoted as $\hat{\gamma}^{split}_j=\hat{\gamma}_j^{split}\left((X^{(2)},y^{(2)}),E\right)$, is constructed using the test set only. To do inference for the marginal LOCO parameter, we use the distribution of the $\hat{\gamma}_j^{split}$ conditional on the training data.
The second one, denoted as $\hat{\gamma}_j^{carved}=\hat{\gamma}_j^{carved}\left((X,y),E\right)$, is constructed using the whole dataset $(X,y)\in\mathbb{R}^{(n_1+n_2)\times p}\times\mathbb{R}^{n_1+n_2}$ we get by combining the training and test sets. To do inference for the marginal LOCO parameter, we use the distribution of $\hat{\gamma}_j^{carved}$ conditional on selecting $E$ in the model selection process for which we only used the training data. We call this approach data carving \citep{fithian2014optimal}.
By conditioning on a smaller part of the training data in the data carving approach compared to data splitting, we have more power for inference.

\begin{itemize}

\item \textbf{Inference for the marginal LOCO via data splitting.}
 The test statistic for $\gamma_j(\mathbb{F},E)$, $j\in E$, is computed using only the second (test) dataset and the added randomization
\begin{equation} \label{eq:loco:data:splitting:estimator}
	\hat{\gamma}_j^{split}\left((X^{(2)}, y^{(2)}),E\right) = \frac{1}{n_2}\sum_{i=1}^{n_2}\left[\left(y^{(2)}_i-{x_{i, E\setminus j}^{(2)}}^\top \bar{\beta}_{E\setminus j}^{(2)}\right)^2-\left(y^{(2)}_i-{x_{i,E}^{(2)}}^{\top}\bar{\beta}_E^{(2)}\right)^2\right]+R_j,
\end{equation}
where $R_j\sim\mathcal{N}(0,\sigma_R^2)$ is independent of the data. ${x^{(2)}_{i,E}}^{\top}$ and ${x_{i,E\setminus j}^{(2)}}^{\top}$ represent the $i$-th row of $X^{(2)}$ restricted to $E$ and $E\setminus j$, respectively.
 $\bar{\beta}^{(2)}_E$ and $\bar{\beta}^{(2)}_{E\setminus j}$ are the OLS estimators computed based on $y^{(2)}\sim X_E^{(2)}$ and $y^{(2)}\sim X_{E\setminus j}^{(2)}$, respectively. Note that the estimators $\bar{\beta}_{E\setminus j}^{(2)}$ and $\bar{\beta}_E^{(2)}$ above are computed on the test data; they can be computed on the training data or even the whole data since we only need these estimators to be consistent for $\beta_{E\setminus j}^*$ and $\beta_E^*$, respectively. 
The inference for $\gamma_j(\mathbb{F},E)$, $j\in E$ is done using the CLT
\begin{equation*}
	\sqrt{n_2}\left(\hat{\gamma}_j^{split}\left((X^{(2)}, y^{(2)}),E\right) -\gamma_j(\mathbb{F},E)\right)
	\overset{d}{\rightarrow}\mathcal{N}\left(0, \sigma_j^{split}(\mathbb{F},E)^2\right)
\end{equation*}
as $n_2\rightarrow\infty$, where the variance $\sigma^{split}_j(\mathbb{F},E)^2$ is estimated using pairs bootstrap.

\item \textbf{Inference for the marginal LOCO via data carving.}
Using the whole dataset, we define the target statistic to be
\begin{equation} \label{eq:loco:carved:estimator}
	\hat{\gamma}_j^{carved}\left((X,y), E\right)=	\frac{1}{n}\sum_{i=1}^n\left(\left(y_i-x^{\top}_{i, E\setminus j}\bar{\beta}_{E\setminus j}\right)^2-\left(y_i-x^{\top}_{i, E,i}\bar{\beta}_E\right)^2\right)+R_j,
\end{equation}
where $R_j\sim\mathcal{N}(0,\sigma_R^2)$. $x_{i,E}^{\top}$ and $x_{i,E\setminus j}^{\top}$ represent the $i$-th row of $X$ restricted to $E$ and $E\setminus j$ respectively. $\bar{\beta}_E$ and $\bar{\beta}_{E\setminus j}$ are OLS estimators computed based on $y\sim X_E$ and $y\sim X_{E\setminus j}$, respectively. Pre-selection, meaning we treat $E$ as fixed, there is a CLT
\begin{equation*}
	\sqrt{n_1+n_2}\left(\hat{\gamma}^{carved}_j\left((X,y),E\right)-\gamma_j(\mathbb{F},E)\right)
	\overset{d}{\rightarrow} \mathcal{N}\left(0,\sigma_j^{carved}(\mathbb{F},E)^2\right) 
\end{equation*}
as $n_1+n_2\rightarrow\infty$ for some variance $\sigma_j^{carved}(\mathbb{F},E)^2$.
The selective $p$-values and intervals are constructed using the estimator in \eqref{eq:loco:carved:estimator} as the target statistics and its post-selective distribution under the null for inference. 

For example, if we use Lasso with fixed penalty value on the training data $(X^{(1)}, y^{(1)})$ to select $E$, we need to condition on this selection event when computing the distribution of the target statistic post-selection. Since we need to re-write the selection event in terms of the target statistic $\hat{\gamma}_j$, $j\in E$, we decompose the data vector $D=\begin{pmatrix}\bar{\beta}^{(1)}_E \\ {{X_{-E}}^{(1)}}^\top\left(y^{(1)}-X_E^{(1)}\bar{\beta}^{(1)}_E\right) \end{pmatrix}$ in terms of this target statistic. Note that all the quantities in the data vector are computed only based on the training data since the selection event depends on the training data only. 
Assuming the joint asymptotic normality of $\hat{\gamma}_j^{carved}$ and $D$, we can re-write $D=N_{\hat{\gamma}_j^{carved}}+\Sigma_{D,\hat{\gamma}_j^{carved}} \cdot \sigma_j^{carved}(\mathbb{F},E)^{-2}\cdot \hat{\gamma}_j^{carved}$, and do inference as in the general framework proposed. We do not write an explicit proof for this joint normality but it follows easily given the previous proofs under mild moment conditions.

\end{itemize}


\begin{remark}
Note that adding the randomization $R_j$ in  \eqref{eq:loco:data:splitting:estimator} is crucial in having $\hat{\gamma}^{split}$ to be asymptotically normal. Similarly, adding the randomization $R_j$ in \eqref{eq:loco:carved:estimator} is crucial to satisfy the requirements of the asymptotic normality of the target statistic and the joint asymptotic normality of the data vector with the target statistic.
The randomization introduced here does not modify the model selection procedure, i.e.~the selected set $E$ does not depend on $R_j$. Hence in this application, the randomization $R_j$ is not within the model selection procedure, but in the inference step.
\end{remark}

\subsection{Simulation example} We take the design matrix $X$ to be of size $n=200$ and $p=50$ with entries i.i.d.~standard Gaussian and normalized to have empirical variance 1. The response $y\sim\mathcal{N}(0, I_n)$, i.e.~a null signal. We use $80\%$ of the data to select the model using plain Lasso with fixed value of $\lambda$. We construct $p$-values and confidence intervals for the  marginal LOCO parameter $\gamma_j(\mathbb{F},E)$, for all selected coefficients $j\in E$, based on both data splitting and data carving. The intervals based on data splitting are constructed using the asymptotic normality of $\hat{\gamma}_j^{split}$, $j\in E$. The carved ones are constructed using the selective sampler.  We sample $(T,\beta_E)$ from density \eqref{eq:randomized:lasso:density}, where $g_{\omega}$ is the normal density coming from the random split; for details see \cite{bootstrap_mv}.

Figure \ref{fig:loco_without_naive} presents the $p$-values for testing whether the marginal LOCO parameter is zero for all $j\in E$. We see that both split and carved $p$-values are valid; however the carved intervals are much shorter than the split intervals as expected since the carved interval leaves more information for inference.

\begin{figure}[h!]
  \centering
    \includegraphics[width=0.60\textwidth]{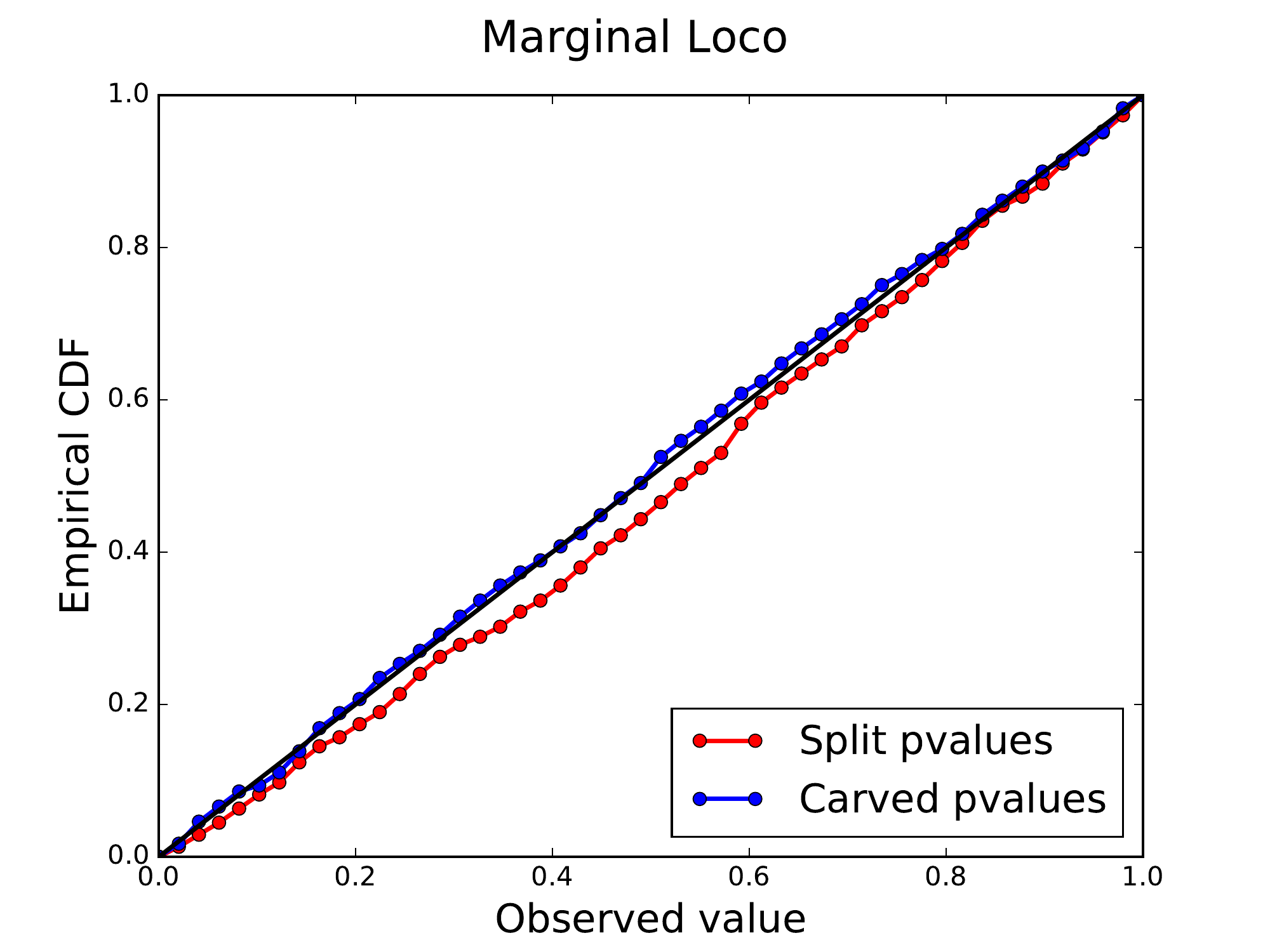}
    \caption{The figure presents the split and carved $p$-values constructed for the marginal LOCO parameter. The average coverage and the average length of intervals via data splitting are 92$\%$ and 0.34, respectively, and for the carved ones  91$\%$ and 0.18, respectively.}
    \label{fig:loco_without_naive}
\end{figure}

\section{Conclusion} 

We have presented a general way of doing selective inference by adjusting for choosing a model based on prediction errors. The examples of our general framework include adjusting inference for choosing penalty level via cross-validation in either randomized or non-randomized Lasso, doing inference after choosing a model based on the minimizer of AIC criteria and after using FS to choose the model with data-dependent stopping time.  It is worth noting that our methods can be applied to any convex loss functions in both the optimization objective and computing the criteria vector, e.g.~$\ell_2$ in both the objective of Lasso and cross-validated error curve.

In the general framework presented in this paper, we need the target statistics and the criteria vector to be asymptotically jointly Gaussian. In some of the examples mentioned, we add randomization to the vector of models' quality values to make it asymptotically jointly Gaussian. This is further used in decomposition of the selection event. In the final example we talked about marginal LOCO parameter and how adding randomization directly to the target might be needed to achieve the joint asymptotically normality.

Although our method is not designed to control the false discovery rate (FDR), we produce the selective $p$-values that can potentially be used in a multiple hypotheses testing framework. However, since our selective $p$-values are in general not independent we do not have theoretical guarantees for these methods, which we leave for future work. 

\section*{Acknowledgment} 
The authors would like to thank Eugene Katsevich for his help in editing this manuscript.

\bibliography{cite}

\begin{thebibliography}{26}
\providecommand{\natexlab}[1]{#1}
\providecommand{\url}[1]{\texttt{#1}}
\expandafter\ifx\csname urlstyle\endcsname\relax
  \providecommand{\doi}[1]{doi: #1}\else
  \providecommand{\doi}{doi: \begingroup \urlstyle{rm}\Url}\fi

\bibitem[Rinaldo et~al.(2016)Rinaldo, Wasserman, G'Sell, Lei, and
  Tibshirani]{rinaldo2016bootstrapping}
Alessandro Rinaldo, Larry Wasserman, Max G'Sell, Jing Lei, and Ryan Tibshirani.
\newblock Bootstrapping and sample splitting for high-dimensional,
  assumption-free inference.
\newblock \emph{arXiv preprint arXiv:1611.05401}, 2016.

\bibitem[Berk et~al.(2013)Berk, Brown, Buja, Zhang, Zhao,
  et~al.]{berk2013valid}
Richard Berk, Lawrence Brown, Andreas Buja, Kai Zhang, Linda Zhao, et~al.
\newblock Valid post-selection inference.
\newblock \emph{The Annals of Statistics}, 41\penalty0 (2):\penalty0 802--837,
  2013.

\bibitem[Taylor and Tibshirani(2015)]{taylor2015statistical}
Jonathan Taylor and Robert~J Tibshirani.
\newblock Statistical learning and selective inference.
\newblock \emph{Proceedings of the National Academy of Sciences}, 112\penalty0
  (25):\penalty0 7629--7634, 2015.

\bibitem[Lee et~al.(2016)Lee, Sun, Sun, and Taylor]{lee2013exact}
Jason~D Lee, Dennis~L Sun, Yuekai Sun, and Jonathan~E Taylor.
\newblock Exact post-selection inference, with application to the lasso.
\newblock \emph{The Annals of Statistics}, 44\penalty0 (3):\penalty0 907--927,
  2016.

\bibitem[Lee and Taylor(2014)]{lee_screening}
Jason~D Lee and Jonathan~E Taylor.
\newblock Exact post model selection inference for marginal screening.
\newblock In \emph{Advances in Neural Information Processing Systems}, pages
  136--144, 2014.

\bibitem[Tibshirani et~al.(2016)Tibshirani, Taylor, Lockhart, and
  Tibshirani]{sequential_post_selection}
Ryan~J. Tibshirani, Jonathan Taylor, Richard Lockhart, and Robert Tibshirani.
\newblock Exact post-selection inference for sequential regression procedures.
\newblock \emph{Journal of the American Statistical Association}, 111\penalty0
  (514):\penalty0 600--620, 2016.

\bibitem[Loftus(2015)]{loftus_cv}
Joshua~R Loftus.
\newblock Selective inference after cross-validation.
\newblock \emph{arXiv preprint arXiv:1511.08866}, 2015.

\bibitem[Tian et~al.(2016{\natexlab{a}})Tian, Bi, and Taylor]{tian2016magic}
Xiaoying Tian, Nan Bi, and Jonathan Taylor.
\newblock Magic: a general, powerful and tractable method for selective
  inference.
\newblock \emph{arXiv preprint arXiv:1607.02630}, 2016{\natexlab{a}}.

\bibitem[Tian et~al.(2016{\natexlab{b}})Tian, Snigdha, Markovic, Bi, and
  Taylor]{selective_sampler}
Xiaoying Tian, Panigrahi Snigdha, Jelena Markovic, Nan Bi, and Jonathan Taylor.
\newblock Selective sampling after solving a convex problem.
\newblock \emph{arXiv preprint arXiv:1609.05609}, 2016{\natexlab{b}}.

\bibitem[Markovic and Taylor(2016)]{bootstrap_mv}
Jelena Markovic and Jonathan Taylor.
\newblock Bootstrap inference after using multiple queries for model selection.
\newblock \emph{arXiv preprint arXiv:1612.07811}, 2016.

\bibitem[Leeb and P{\"o}tscher(2006{\natexlab{a}})]{leeb2006can}
Hannes Leeb and Benedikt~M P{\"o}tscher.
\newblock Can one estimate the conditional distribution of post-model-selection
  estimators?
\newblock \emph{The Annals of Statistics}, 34\penalty0 (5):\penalty0
  2554--2591, 2006{\natexlab{a}}.

\bibitem[Leeb and P{\"o}tscher(2006{\natexlab{b}})]{leeb2006performance}
Hannes Leeb and Benedikt~M P{\"o}tscher.
\newblock Performance limits for estimators of the risk or distribution of
  shrinkage-type estimators, and some general lower risk-bound results.
\newblock \emph{Econometric Theory}, 22\penalty0 (01):\penalty0 69--97,
  2006{\natexlab{b}}.

\bibitem[Freedman et~al.(1981)]{freedman1981bootstrapping}
David~A Freedman et~al.
\newblock Bootstrapping regression models.
\newblock \emph{The Annals of Statistics}, 9\penalty0 (6):\penalty0 1218--1228,
  1981.

\bibitem[Buja et~al.(2014)Buja, Berk, Brown, George, Pitkin, Traskin, Zhang,
  Zhao, and White]{buja2014conspiracy}
A~Buja, R~Berk, L~Brown, E~George, E~Pitkin, M~Traskin, K~Zhang, L~Zhao, and
  Dedicated To~Halbert White.
\newblock A conspiracy of random predictors and model violations against
  classical inference in regression.
\newblock \emph{arXiv preprint}, 2014.

\bibitem[Homrighausen and McDonald(2013)]{cv_risk_consistency}
Darren Homrighausen and Daniel~J McDonald.
\newblock The lasso, persistence, and cross-validation.
\newblock In \emph{ICML (3)}, pages 1031--1039, 2013.

\bibitem[Dudoit and van~der Laan(2005)]{dudoit2005asymptotics}
Sandrine Dudoit and Mark~J van~der Laan.
\newblock Asymptotics of cross-validated risk estimation in estimator selection
  and performance assessment.
\newblock \emph{Statistical Methodology}, 2\penalty0 (2):\penalty0 131--154,
  2005.

\bibitem[Friedman et~al.(2001)Friedman, Hastie, and
  Tibshirani]{friedman2001elements}
Jerome Friedman, Trevor Hastie, and Robert Tibshirani.
\newblock \emph{The elements of statistical learning}, volume~1.
\newblock Springer series in statistics Springer, Berlin, 2001.

\bibitem[Bubeck et~al.(2015)Bubeck, Eldan, and Lehec]{bubeck2015sampling}
S{\'e}bastien Bubeck, Ronen Eldan, and Joseph Lehec.
\newblock Sampling from a log-concave distribution with {P}rojected {L}angevin
  {M}onte {C}arlo.
\newblock \emph{arXiv preprint arXiv:1507.02564}, 2015.

\bibitem[Tian and Taylor(2015)]{tian2015selective}
Xiaoying Tian and Jonathan~E Taylor.
\newblock Selective inference with a randomized response.
\newblock \emph{arXiv preprint arXiv:1507.06739}, 2015.

\bibitem[Barber and Cand{\`e}s(2015)]{original_knockoffs}
Rina~Foygel Barber and Emmanuel~J Cand{\`e}s.
\newblock Controlling the false discovery rate via knockoffs.
\newblock \emph{The Annals of Statistics}, 43\penalty0 (5):\penalty0
  2055--2085, 2015.

\bibitem[Fithian et~al.(2014)Fithian, Sun, and Taylor]{fithian2014optimal}
William Fithian, Dennis Sun, and Jonathan Taylor.
\newblock Optimal inference after model selection.
\newblock \emph{arXiv preprint arXiv:1410.2597}, 2014.

\bibitem[Lehmann and Romano(2006)]{TSH}
Erich~L Lehmann and Joseph~P Romano.
\newblock \emph{Testing statistical hypotheses}.
\newblock Springer Science \& Business Media, 2006.

\bibitem[Candes et~al.(2016)Candes, Fan, Janson, and Lv]{candes2016panning}
Emmanuel Candes, Yingying Fan, Lucas Janson, and Jinchi Lv.
\newblock Panning for gold: Model-free knockoffs for high-dimensional
  controlled variable selection.
\newblock \emph{arXiv preprint arXiv:1610.02351}, 2016.

\bibitem[Benjamini and Hochberg(1995)]{benjamini1995controlling}
Yoav Benjamini and Yosef Hochberg.
\newblock Controlling the false discovery rate: a practical and powerful
  approach to multiple testing.
\newblock \emph{Journal of the royal statistical society. Series B
  (Methodological)}, pages 289--300, 1995.

\bibitem[Barber and Candes(2016)]{high_dim_knockoffs}
Rina~Foygel Barber and Emmanuel~J Candes.
\newblock A knockoff filter for high-dimensional selective inference.
\newblock \emph{arXiv preprint arXiv:1602.03574}, 2016.

\bibitem[Friedman et~al.(2015)Friedman, Hastie, and R.]{friedmantibshirani}
J~Friedman, T~Hastie, and Tibshirani R.
\newblock glmnet: Lasso and elastic-net regularized generalized linear models,
  2015.

\end{thebibliography}
\bibliographystyle{imsart-nameyear}
\clearpage	

\appendix
\section{Proofs} \label{sec:proofs}
$$\textnormal{ }$$
\begin{proof-of-theorem}{\ref{thm:nonrandomized:pivot}}
	Recall that $\Phi$ denotes the asymptotic normal distribution of $(T,D)$, the RHS of \eqref{eq:D:T:gaussian}, and $\Phi^*$ the distribution of $(Z_T,Z_{\widetilde{D}})$ conditional on $\widetilde{A}Z_{\widetilde{D}}\leq\tilde{a}$.
	Using the CLT assumption, we have for all $t\in\mathbb{R}^{\textnormal{dim}(T)}$ and $\tilde{d}\in\mathbb{R}^{\textnormal{dim}(\widetilde{D})}$
	\begin{align*}
		&\mathbb{F}_n^*\left\{T\leq t, \widetilde{D}\leq \tilde{d}\right\} =\mathbb{F}_n\left\{T\leq t,\widetilde{D}\leq\tilde{d}\:\Big|\:\widetilde{A}\widetilde{D}\leq\tilde{a}_n\right\}
		= \frac{\mathbb{F}_n\left\{T\leq t,\widetilde{D}\leq \tilde{d},\widetilde{A}\widetilde{D}\leq \tilde{a}_n\right\}}{\mathbb{F}_n\left\{\widetilde{A}\widetilde{D}\leq \tilde{a}_n\right\}} \\
		&\overset{n\rightarrow\infty}{\rightarrow}\frac{\Phi\left\{Z_T\leq t,Z_{\widetilde{D}}\leq \tilde{d}, \widetilde{A} Z_{\widetilde{D}}\leq \tilde{a}\right\}}{\Phi\left\{AZ_{\widetilde{D}}\leq \tilde{a}\right\}}=\Phi^*\left\{Z_T\leq t, Z_{\widetilde{D}}\leq \tilde{d}\right\},
	\end{align*}
	where we used \cite[Theorem 11.2.9]{TSH} and \cite[Corollary 11.2.3]{TSH} to show the convergence step.
	 This implies 
	\begin{equation*}
		\mathbb{F}_n^*\overset{d}{\rightarrow}\Phi^* 
	\end{equation*}
	as $n\rightarrow\infty$, or equivalently under $(T,\widetilde{D})\sim\mathbb{F}_n$ and $(Z_T,Z_{\widetilde{D}})\sim\Phi$
	\begin{equation*}
		\begin{pmatrix} T\\ D \end{pmatrix}	 \Big|\: \widetilde{A}\widetilde{D}\leq \tilde{a}_n \overset{d}{\rightarrow} \begin{pmatrix} Z_T \\ Z_{\widetilde{D}} \end{pmatrix} \Big |\: \widetilde{A}Z_{\widetilde{D}}\leq\tilde{a}
	\end{equation*}
	as $n\rightarrow\infty$.
	
	Turning to the pivot, $\mathcal{P}\left((t,\tilde{d});\widetilde{A},\tilde{a}\right)$ is continuous in $(t,\tilde{d},\tilde{a})$ for $\widetilde{A}\cdot \tilde{d}\leq \tilde{a}$. Thus, under $\Phi^*$, $\mathcal{P}\left((Z_T,Z_{\widetilde{D}});\widetilde{A},\tilde{a}\right)$ is continuous with probability 1 in $(Z_T,Z_{\widetilde{D}},\tilde{a})$. By Continuous Mapping Theorem, using the continuity of the pivot with the weak convergence result above we have under $(T,\widetilde{D})\sim\mathbb{F}_n^*$ and $(Z_T,Z_{\widetilde{D}})\sim\Phi^*$
	\begin{equation*}
		\mathcal{P}\left((T,\widetilde{D}); \widetilde{A},\tilde{a}_n\right)\overset{d}{\rightarrow}\mathcal{P}\left((Z_T,Z_{\widetilde{D}});\widetilde{A},\tilde{a}\right)\sim\textnormal{Unif}(0,1)
	\end{equation*}
	as $n\rightarrow\infty$.
	
\end{proof-of-theorem}

\begin{proof-of-lemma}{\ref{lemma:consistency}}
Given $\epsilon>0$, we have
\begin{equation*}
\begin{aligned}
	\mathbb{F}_n^*\left\{\left\|\hat{\xi}-\xi\right\|_2\geq\epsilon\right\}
	&=\mathbb{F}_n\left\{\left\|\hat{\xi}-\xi\right\|_2\geq\epsilon \:\Big|\: \widetilde{D}\in\mathcal{S}_{\widetilde{D}}\right\}
	=\frac{\mathbb{F}_n\left\{\left\|\hat{\xi}-\xi\right\|_2\geq \epsilon,\widetilde{D}\in\mathcal{S}_{\widetilde{D}}\right\}}{\mathbb{F}_n\left\{\widetilde{D}\in\mathcal{S}_{\widetilde{D}}\right\}},
\end{aligned}
\end{equation*}	
implying
\begin{equation} \label{eq:consistency:inequality}
	\mathbb{F}_n^*\left\{\left\|\hat{\xi}-\xi\right\|_2\geq\epsilon\right\}
	\leq \frac{\mathbb{F}_n\left\{\left\|\hat{\xi}-\xi\right\|_2\geq \epsilon\right\}}{\mathbb{F}_n\left\{\widetilde{D}\in\mathcal{S}_{\widetilde{D}}\right\}}.
\end{equation}
Since by the convergence assumption we have that the denominator above converges to a positive constant, i.e.
\begin{equation*}
	\underset{n\rightarrow\infty}{\lim}\mathbb{F}_n\left\{\widetilde{D}\in\mathcal{S}_{\widetilde{D}}\right\} = \Phi\left\{\widetilde{A}Z_{\widetilde{D}}\leq \tilde{a}\right\}>0.
\end{equation*}
By the consistency of $\hat{\xi}$ pre-selection, we have that the numerator in \eqref{eq:consistency:inequality} converges to zero in the limit as $n$ tends to infinity. Combining these two observations, we have 
\begin{equation*}
	\underset{n\rightarrow\infty}{\lim}\mathbb{F}_n^*\left\{\left\|\hat{\xi}-\xi\right\|_2\geq\epsilon\right\}=0.
\end{equation*}
\end{proof-of-lemma}


\begin{proof-of-theorem}{\ref{thm:test:error:fixed:p}}
	We write the decomposition
\begin{equation*}
\begin{aligned}
	&\left\|y^{test}-X^{test}\hat{\beta}^{train}(\lambda)\right\|_2^2 = \left\|y^{test}-X^{test}\beta_0(\lambda)+X^{test}\beta_0(\lambda)-X^{test}\hat{\beta}^{train}(\lambda)\right\|_2^2\\
	&=\left\|y^{test}-X^{test}\beta_0(\lambda)\right\|_2^2
	+2{\epsilon^{test}(\lambda)}^\top\left(X^{test}\beta_0(\lambda)-X^{test}\hat{\beta}^{train}(\lambda)\right) \\
	&\qquad+\left\|X^{test}\beta_0(\lambda)-X^{test}\hat{\beta}^{train}(\lambda)\right\|_2^2,
\end{aligned}
\end{equation*}
where $\epsilon^{test}(\lambda)=y^{test}-X^{test}\beta_0(\lambda)$, and analyze each of the three terms separately.

\begin{enumerate}

\item  Using the assumption that $\textnormal{Var}\left((y_1-x_1^\top\beta_0(\lambda))^2\right)<\infty$, by the CLT we have that
\begin{equation*}
	\frac{1}{\sqrt{n_2}}\left(\left\|y^{test}-X^{test}\beta_0(\lambda)\right\|_2^2-\mathbb{E}\left[\left\|y^{test}-X^{test}\beta_0(\lambda)\right\|_2^2\right]\right) \overset{d}{\rightarrow}\mathcal{N}(0,\sigma^2(\lambda))
\end{equation*}
as $n_2\rightarrow\infty$.

\item Using the Cauchy-Schwartz inequality, we have
\begin{equation} \label{eq:inequality:2}
\begin{aligned}
	&\frac{1}{\sqrt{n_2}}\left|\left({X^{test}}^\top \epsilon^{test}(\lambda)\right)^{\top}\left(\hat{\beta}^{train}(\lambda)-\beta_0(\lambda)\right)\right|\\
	&\qquad \leq \frac{1}{\sqrt{n_2}}\left\|{X^{test}}^{\top}\epsilon^{test}(\lambda)\right\|_2\left\|\hat{\beta}^{train}(\lambda)-\beta_0(\lambda)\right\|_2.
\end{aligned}
\end{equation}
Using that $\mathbb{E}\left[\left\|x_1\left(y_1-x_1^\top\beta_0(\lambda)\right)\right\|_2^2\right]<\infty$, by LLN we have  $\frac{1}{n_2}\left\|{X^{test}}^\top\epsilon^{test}(\lambda)\right\|_2^2=O_P(1)$ as $n_2\rightarrow\infty$. By the consistency of $\hat{\beta}^{train}(\lambda)$, we have $\left\|\hat{\beta}^{train}(\lambda)-\beta_0(\lambda)\right\|_2\rightarrow 0$ as $n_1\rightarrow\infty$. We conclude the RHS in \eqref{eq:inequality:2} is $o_P(1)$ as $n_1,n_2\rightarrow\infty$.

\item We have
\begin{equation} \label{inequality:3}
	\frac{1}{\sqrt{n_2}}\left\|X^{test}\beta_0(\lambda)-X^{test} \hat{\beta}^{train}(\lambda)\right\|_2^2 
	 \leq \frac{1}{\sqrt{n_2}}\sum_{i=1}^{n_2}\left\|x^{test}_i\right\|_2^2\left\|\beta_0(\lambda)-\hat{\beta}^{train}(\lambda)\right\|_2^2,
\end{equation}
where $x_i^{test}\in\mathbb{R}^p$, $i=1,\ldots, n_2$, are the rows of the test design matrix $X^{test}$. By the consistency assumption, we have $\sqrt{n_1}\left\|\hat{\beta}^{train}(\lambda)-\beta_0(\lambda)\right\|_2^2\rightarrow 0$ as $n_1\rightarrow\infty$. Using the assumption that $\mathbb{E}\left[\left\|x_1\right\|_2^2\right]<\infty$, we have that $\frac{1}{n_2}\sum_{i=1}^{n_2}\left\|x_i^{test}\right\|_2^2=O_P(1)$ as $n_2\rightarrow\infty$ by the LLN. Hence we conclude that the RHS in \eqref{inequality:3} is $o_P(1)$ as $n_1, n_2\rightarrow\infty$.

\end{enumerate}

\end{proof-of-theorem}

\begin{proof-of-corollary}{\ref{cor:randomized:cv:gaussian}}
	From the proof of Theorem \ref{thm:test:error:fixed:p} we have that for each $k=1,\ldots, K$ and each $\lambda\in\Lambda$ the following holds
	\begin{equation*}
	\begin{aligned}
		\frac{1}{\sqrt{n_k}}\left(\left\|y^k-X^k\hat{\beta}^{-k}(\lambda)\right\|_2^2 -n_k\mu(\lambda)\right)
		&=\frac{1}{\sqrt{n_k}}\left(\left\|\epsilon^k(\lambda)\right\|_2^2-n_k\mu(\lambda)\right)+o_P(1)\\
		&=\frac{1}{\sqrt{n_k}}\left(\sum_{i=1}^{n_k}\epsilon^k_i(\lambda)^2-n_k\mu(\lambda)\right)+o_P(1),
	\end{aligned}	
	\end{equation*}
	where $\epsilon^k(\lambda)=(\epsilon_1^k(\lambda),\ldots, \epsilon_{n_k}^k(\lambda))=y^k-X^k\beta_0(\lambda)\in\mathbb{R}^{n_k}$. Hence, the random variables $\epsilon^k_i(\lambda)$ are i.i.d.~across $i=1,\ldots, n_k$ and across $k=1,\ldots, K$ for each $\lambda\in\Lambda$. Summing the above across folds $k=1,\ldots, K$, we have for each $\lambda\in\Lambda$
	\begin{equation*}
	\begin{aligned}
		\sum_{k=1}^K\frac{1}{n_k}\left(\left\|y^k-X^k\hat{\beta}^{-k}(\lambda)\right\|_2^2-n_k\mu(\lambda)\right)
		&=\sum_{k=1}^K\left(\frac{1}{n_k}\left\|\epsilon^k(\lambda)\right\|_2^2-\mu(\lambda)+o_P\left(\frac{1}{\sqrt{n_k}}\right)\right).
	\end{aligned}
	\end{equation*}
	This implies
	\begin{equation*}
	\begin{aligned}
	Err_R(\lambda) &= \sum_{k=1}^{K}\left(\frac{1}{n_k}\left\|y^k-X^k\hat{\beta}^{-k}(\lambda)\right\|_2^2+\frac{1}{\sqrt{n_k}}R^{k,\lambda}\right) \\
	&=\sum_{k=1}^K\frac{1}{n_k}\left\|\epsilon^k(\lambda)\right\|_2^2+o_p\left(\frac{1}{\sqrt{n}}\right)+\sum_{k=1}^K\frac{1}{\sqrt{n_k}}R^{k,\lambda}.
	\end{aligned}
	\end{equation*}
	By the CLT, $Err_R=\left(Err_R(\lambda_1), \ldots,Err_R(\lambda_L)\right)$, properly scaled and centered, is asymptotically a jointly normal vector.
	
\end{proof-of-corollary}


\begin{proof-of-corollary}{\ref{cor:joint:normality}}
	Denote $\mathbb{E}\left[\left(\frac{1}{n}X_E^\top X_E\right)^{-1}\right]=M_E$ and $\mathbb{E}\left[\frac{X_{-E}^\top X_E}{n}\right]=M_{-E}$. Let $x_{i,E}$, $i=1,\ldots,n$, denote the rows of $X_E$ and $x_{i,-E}$, $i=1,\ldots, n$, the rows of $X_{-E}$.
	We have
	\begin{equation*}
	\begin{aligned}
	&\sqrt{n}\left(\bar{\beta}_E-\beta_E^*\right) = \left(\frac{1}{n}X_E^\top X_E\right)^{-1}\frac{1}{\sqrt{n}}X_E^\top\left(y-X_E\beta_E^*\right) \\
	&=\left(\left(\frac{1}{n}X_E^\top X_E\right)^{-1}-M_E\right)\frac{1}{\sqrt{n}}X_E^\top\left(y-X_E\beta_E^*\right)  +M_E\frac{1}{\sqrt{n}}X_E^\top\left(y-X_E\beta_E^*\right) \\
	&= o_P(1)O_P(1)+
	M_E\frac{1}{\sqrt{n}}\sum_{i=1}^nx_{i,E}^\top\left(y_i-x_{i,E}^\top\beta_E^*\right).
	\end{aligned}
	\end{equation*}
	
	Since
	\begin{equation*}
	\begin{aligned}
	&\frac{1}{\sqrt{n}}\left(X_{-E}^\top\left(y-X_E\bar{\beta}_E\right)-X_{-E}^\top\left(y-X_E\beta_E^*\right)\right) 
	= -\frac{X_{-E}^\top X_E}{n}\sqrt{n}\left(\bar{\beta}_E-\beta_E^*\right) \\
	&=-\left(\frac{X_{-E}^\top X_E}{n}-M_{-E}\right)\sqrt{n}\left(\bar{\beta}_E-\beta_E^*\right)-M_{-E}\sqrt{n}\left(\bar{\beta}_E-\beta_E^*\right)\\
	&=o_P(1)O_P(1)-M_{-E}\left(o_P(1)+M_E\frac{1}{\sqrt{n}}\sum_{i=1}^nx_{i,E}^\top\left(y_i-x_{i,E}^\top\beta_E^*\right)\right),
	\end{aligned}
	\end{equation*}
	we have
	\begin{equation*}
	\begin{aligned}
	&\frac{1}{\sqrt{n}}\left(X_{-E}^\top\left(y-X_E\bar{\beta}_E\right)-\mathbb{E}\left[X_{-E}^\top\left(y-X_E\beta_E^*\right)\right]\right)\\
	&=o_P(1)+\frac{1}{\sqrt{n}}\left(\sum_{i=1}^nx_{i,-E}^\top\left(y_i-x_{i,E}^\top\beta_E^*\right)-\mathbb{E}\left[X_{-E}^\top\left(y-X_E\beta_E^*\right)\right]\right) \\
	&\qquad -M_{-E}M_E\frac{1}{\sqrt{n}}\sum_{i=1}^nx_{i,E}^\top\left(y_i-x_{i,E}^\top\beta_E^*\right)
	\end{aligned}
	\end{equation*}
	Hence 
	\begin{equation*}
	\sqrt{n}\left(\begin{pmatrix}
		\bar{\beta}_E \\ \frac{1}{n}X_{-E}^\top(y-X_E\bar{\beta}_E)
	\end{pmatrix} - \begin{pmatrix} \beta_E^*\\ \frac{1}{n}\mathbb{E}\left[X_{-E}^\top\left(y-X_E\beta_E^*\right)\right]\end{pmatrix}\right)
	\end{equation*} 
	is asymptotically the sum of i.i.d.~random variables with the leftover term $o_P(1)$. From the proof of Corollary \ref{cor:randomized:cv:gaussian}, we know that $Err_R$ is also represented in such a form. Thus the joint normality follows.

\end{proof-of-corollary}


\begin{proof-of-theorem}{\ref{thm:AIC}}
	Decomposing the $\|y-X_{E_i}\bar{\beta}_{E_i}\|$, $i=1,2$, into a difference of two terms
	\begin{equation*}
	\begin{aligned}
	\left\|y-X_{E_i}\bar{\beta}_{E_i}\right\|_2^2
	&=\left\|y-X_{E_i}\beta_{E_i}^*\right\|_2^2
	+2(y-X_{E_i}\beta_{E_i}^*)^\top X_{E_i}(\beta_{E_i}^*-\bar{\beta}_{E_i})
	+\left\|X_{E_i}\beta_{E_i}^*-X_{E_i}\bar{\beta}_{E_i}\right\|_2^2 \\
	&=\left\|y-X_{E_i}\beta_{E_i}^*\right\|_2^2-\left\|X_{E_i}\beta_{E_i}^*-X_{E_i}\bar{\beta}_{E_i}\right\|_2^2,
	\end{aligned}
	\end{equation*}
	we analyze each of them separately.
	\begin{enumerate}
	\item There is a CLT
		\begin{equation*}
			\frac{1}{\sqrt{n}}\left(\left\|y-X_{E_i}\beta_{E_i}^*\right\|_2^2-\mathbb{E}\left[\|\epsilon_{E_i}\|_2^2\right]\right) \overset{d}{\rightarrow}\mathcal{N}\left(0,\sigma_{E_i}^2\right)
		\end{equation*}
		as $n\rightarrow\infty$.
	\item We have
	\begin{equation} \label{eq:AIC:ineq:3}
		\frac{1}{\sqrt{n}}\left\|X_{E_i}(\bar{\beta}_{E_i}-\beta_{E_i}^*)\right\|_2^2 \leq \frac{1}{\sqrt{n}}\sum_{i=1}^n\|x_{i, E_i}\|_2^2\left\|\bar{\beta}_{E_i}-\beta_{E_i}^*\right\|_2^2,
	\end{equation}
	where $x_{i,E_i}\in\mathbb{R}^{|E_i|}$ are the sub-rows of $X$ corresponding to the columns in $E_i$.
	By the LLN, we have $\frac{1}{n}\sum_{i=1}^n\|x_{i,E_i}\|_2^2=O_P(1)$. By the consistency assumption $\sqrt{n}\left\|\bar{\beta}_{E_i}-\beta_{E_i}^*\right\|_2^2\rightarrow 0 $ as $n\rightarrow\infty$. We conclude the RHS of \ref{eq:AIC:ineq:3} is $o_P(1)$ as $n\rightarrow\infty$.
	\end{enumerate}
	This proves for $i=1,2$
	\begin{equation*}
	\frac{1}{\sqrt{n}}\left(\|y-X_{E_i}\bar{\beta}_{E_i}\|_2^2-\mathbb{E}\left[\|\epsilon_{E_i}\|_2^2\right]\right)	 = \frac{1}{\sqrt{n}}\left(\sum_{i=1}^n\epsilon_{i,E_i}^2-\mathbb{E}\left[\|\epsilon_{E_i}\|_2^2\right]\right)+o_P(1),
	\end{equation*}
	where $\epsilon_{E_i}=(\epsilon_{1,E_i},\ldots, \epsilon_{n,E_i})\in\mathbb{R}^n$ consists of i.i.d.~coordinates, implying the conclusion.
\end{proof-of-theorem}


\section{Projection step in sampling} \label{sec:sampling}

We use projected Langevin Monte Carlo with guarantees in \cite{bubeck2015sampling} to sample from a log-concave density. This MC has been used in other randomized selective inference problems in \cite{selective_sampler, bootstrap_mv}. We omit the details of the sampler that can be found in these works but focus on the differences involving a new optimization variable $Err_R$. At every step of the sampler $Err_R\in\mathbb{R}^L$ is projected onto the cone specifying that the smallest $K$ values of $Err_R$ are achieved at the fixed coordinates $\{l_1, \ldots, l_K^*\}$. In this section, we describe the computation on this projection.

For $K=1$, we have the following problem: for a given $z\in\mathbb{R}^L$, solve 
\begin{equation*}
	\underset{x\in\mathbb{R}^L}{\textnormal{minimize }} \|x-z\|_2^2 \;\; \textnormal{ such that }x_j\geq x_{l_1^*}\;\forall j:1\leq j\leq L,
\end{equation*}
where $l_1^*$ is a given index, $l_1^*\in\{1,\ldots, L\}$. This is equivalent to projecting $z$ on the space consisting of all vectors in $\mathbb{R}^L$ whose minimum is achieved at coordiante $l_1^*$. Assume w.l.o.g.~$l_1^*=1$. To solve this problem,  we first fix $x_1$ and see that the other optimizing values for other coordinates, $j=2,\ldots, L$, are
\begin{equation*}
	x_j^*(x_1) = \textnormal{max}\{x_1,z_j\}.
\end{equation*}
Hence, the problem becomes to minimize the convex function
\begin{equation*}
	h(x_1)=(x_1-z_1)^2+\sum_{l=2}^L(x_1-z_l)^2\mathbb{I}_{\{x_1\geq z_l\}}.
\end{equation*}
Derivative of $h$ w.r.t.~$x_1$ is
\begin{equation*}
	h'(x_1) = 2(x_1-z_1)+2\sum_{l=2}^L(x_1-z_l)\mathbb{I}_{\{x_1\geq z_l\}}.
\end{equation*}
We evaluate $h'(x_1)$ at all $z_l$, $l=1,\ldots, L$, to find an interval over which $h'$ crosses zero. Having found such an interval, the root $x_1^*$ will be the sample average of $z_1$ and all non-zero summands in $h$.


For a general $K\leq L$, we have the following problem: for a given $z\in\mathbb{R}^L$, solve
\begin{equation*}
	\underset{x\in\mathbb{R}^L}{\textnormal{minimize }} \|x-z\|_2^2 \;\;\textnormal{ such that }\;\; x_j\leq x_l \;\forall j\in \Lambda^*, \forall l\in-\Lambda^*.
\end{equation*}
Keeping $v$, $\textnormal{max}\{x_j:j\in L^*\}\leq v\leq\textnormal{min}\{x_j: j\in -L^*\}$, fixed, the optimizing values for $j=1,\ldots, L$ are
\begin{equation*}
	x_j^*(v) = \left\{ \begin{matrix} \textnormal{min}\{v,z_j\} & j\in\Lambda^*  
	\\ \textnormal{max}\{v,z_j\} & j\not\in\Lambda^* \end{matrix}\right..
\end{equation*}
The problem becomes to minimize
\begin{equation*}
	h(v) = \sum_{j\in\Lambda^*}(v-z_j)^2\mathbb{I}_{\{v\leq z_j\}} +\sum_{j\in -\Lambda^*}(v-z_j)^2\mathbb{I}_{\{v\geq z_j\}}.
\end{equation*}
We find $v$ as above.


\section{Randomized forward-stepwise with data dependent number of steps - details} \label{sec:FS:details}

This section provides the details of the randomized forward-stepwise procedure with the data dependent number of steps chosen via the randomized FS criteria as described in Section \ref{sec:FS:main}.
We combine the selection events coming from the randomized forward-stepwise (FS) algorithm with choosing a data dependent $L$, the number of steps in FS.

\subsection{FS with fixed $L$}

The description of randomized forward-stepwise with fixed $L$ is given in \cite{selective_sampler, bootstrap_mv} and we revise it here for completeness. 
The data generating mechanism on $(X,y)\sim\mathbb{F}_n^n$ is as in the previous examples. In the $L$ steps of forward stepwise, the selection event is characterized by a sequence of indices $\bb j=(j_1,\ldots, j_L)$ with their corresponding signs $\bb s=(s_1,\ldots, s_L)$ that enter the model in that particular order, forming an active set at step $L$. Denote the active set at step $l$ as $E_l=\{j_1,j_2,\ldots, j_l\}$ for all $l=1,\ldots, L$.
At the $l$-th step the randomized forward stepwise solves the following program
\begin{equation} \label{eq:fs:objective:original}
	\hat{\eta}_l=\textnormal{arg}\underset{\eta\in \mathcal{B}_l}{\max}\:\eta^\top \left(X_{-E_{l-1}}^\top P_{E_{l-1}}^{\perp}y+ \omega_l\right), \;\; (X,y)\times \omega_l\sim\mathbb{F}_n^n\times\mathbb{F}_{\omega_l},
\end{equation}
where $\mathcal{B}_l=\{\eta \in\mathbb{R}^{p-l+1}: \|\eta\|_1\leq 1\},$ and $P_{E_{l-1}}^{\perp} y$ is the residual left after projecting $y$ onto $X_{E_{l-1}}$. $\mathbb{F}_{\omega_l}$ is pre-specified distribution of the randomization $\omega_l$ with known density $g_l$.

The selection event of interest is given by conditioning on the sign and the index on the non-zero coordinate of the solution $\hat{\eta}_l$ for each $l=1,\ldots, L$. We want to sample from the density of the data and the randomization conditional on this selection event. The randomization reconstruction map for the $l$-th step, from the sub-gradient equation is given by
\begin{equation*}
	\omega_l(y,z_l)=-X_{-E_{l-1}}^\top P_{E_{l-1}}^{\perp}y+z_l,
\end{equation*}
where, sub-differential $z_l\in \mathbb{R}^{p-l+1}$ from the $l$-th step is restricted to the normal cone
$z_l\in \partial I_{\mathcal{B}_l}(\hat{\eta}_l)$ (see \cite{selective_sampler}).
The selective density of $((X,y), z_1,\ldots, z_L)$ is then proportional to
\begin{equation} \label{eq:fs:density}
	\left(\prod_{i=1}^n f_n (x_i,y_i)\right) \cdot \prod_{l=1}^L g_l\left(z_l-X_{-E_{l-1}}^\top P_{E_{l-1}}^{\perp} y\right),
\end{equation}
supported on $\mathbb{R}^{n\times p}\times\mathbb{R}^n \times \prod_{l=1}^L \partial I_{\mathcal{B}_l}(\hat{\eta}_l)$, where $f_n$ denotes the density of $\mathbb{F}_n$.

After doing $L$ steps of forward stepwise, an analyst looks at the sequence $\{E_l\}_{l=1}^L$ and chooses model $E$ in whichever way she wants. The goal is to inference for the population OLS parameters $\beta_E^*$. 
As in the LASSO example, we simplify the sampling above since we are interested in testing a particular parameter. First note that
$X_{-E_{l-1}}^\top P_{E_{l-1}}^{\perp} y$ can be expressed as $Q_l\cdot  X^\top y$, where
$Q_l= \left[\begin{matrix} X_{-E_{l-1}}^\top X_{E_{l-1}}\left(X_{E_{l-1}}^\top X_{E_{l-1}}\right)^{-1} & -I_{p-(l-1)}	\end{matrix} \right].$
Using the asymptotic normality of $D=\begin{pmatrix} \bar{\beta}_E \\ X_{-E}^\top(y- X_E\bar{\beta}_E) \end{pmatrix}$, the sampling density of $(D, z_1,\ldots, z_L)$ is proportional to
\begin{equation*}
	\phi_{(\mu_{\bm D}, \Sigma_D)}(D)\cdot \prod_{l=1}^L g_l \left(z_l+M_lD\right),
\end{equation*}
and supported on $\mathbb{R}^p\times\prod_{l=1}^L \partial I_{\mathcal{B}_l}(\hat{\eta}_l)$, where $M_l = - Q_l\begin{pmatrix} X_E^\top X_E & 0\\ X_{-E}^\top X_E & I_{p-|E|}\end{pmatrix}$.

Given the parameter of interest $\theta=\theta(\mathbb{F}_n, E)$ and the corresponding test statistic $T$, we assume that $(T,D)$ is jointly asymptotically normal vector.
Using the decomposition $D=\Sigma_{D,T}\Sigma_T^{-1}T+N_D$ as in the previous examples, the sampling density of $T$ along with optimization variables is proportional to 
\begin{equation} \label{eq:fs:density:plugin:clt}
	\phi_{(\theta,\Sigma_T)}(T) \cdot \prod_{l=1}^L g_l \left(z_l+ M_l N_D+M_l\Sigma_{D,T}\Sigma_T^{-1}T\right)
\end{equation}
with the restriction $(z_1,\ldots, z_L)\in\prod_{l=1}^L \partial I_{\mathcal{B}_l}(\hat{\eta}_l)$. 

\subsection{FS with data dependent $L$}

We take into account that the number of steps $L$ is chosen in a data dependent way. Given the randomized FS errors $Err_{R,l}$, $1\leq l\leq \min\{n,p\}$, $L$ is chosen such that
\begin{equation*}
	L=\min\left\{2\leq l\leq\min\{n,p\} \;\textnormal{ s.t.}\; Err_{R,l}\leq \eta \cdot Err_{R,l-1} \right\},
\end{equation*}
where $\eta$ is a constant. Conditioning on $L$, the constraint coming from choosing $L$ is equivalent to requiring that $Err_R=\left(Err_{R,1},\ldots, Err_{R,L}\right)\in\mathbb{R}^L$ satisfies
\begin{equation} \label{eq:FS:curve:constraint}
	B_L	\cdot Err_R\leq 0,
\end{equation}
for a fixed matrix $B_L$ depending on $L$. Recall that $Err_R=Err+R_1+R_2=Err_R^{(1)}+R_2$ with $R_2\sim\mathbb{F}_{R_2}$ with density $g_{R_2}$. 
As this selection event is written in terms of $Err_R$, we write it in terms of the randomization $R_2$ and the target statistic $T$ by decomposing $Err_R^{(1)}=N_{Err_R^{(1)}}+\Sigma_{Err_R^{(1)},T}\Sigma_T^{-1}T$ as
\begin{equation*}
	R_2=Err_R-N_{Err_R^{(1)}}-\Sigma_{Err_R^{(1)},T}\Sigma_T^{-1}T
\end{equation*}
with the restriction on $Err_R$ as in \eqref{eq:FS:curve:constraint}.

Combining the selection event coming from looking at the selected predictors via FS together with choosing $L$ as above, we write the selective density on $(T, z_1,\ldots, z_L, Err_R)$ as proportional to
\begin{equation*}
\begin{aligned}
	\phi_{(\theta,\Sigma_T)}(T) &\cdot \prod_{l=1}^L g_l \left(z_l+ M_l N_D+M_l\Sigma_{D,T}\Sigma_T^{-1}T\right)\\
	&\cdot g_{R_2}\left(Err_R-N_{Err_R^{(1)}}-\Sigma_{Err_R^{(1)},T}\Sigma_T^{-1}T\right)
\end{aligned}	
\end{equation*}
with the sub-gradients $(z_1,\ldots, z_L)$ restricted as above to a product of normal cones and $Err_R$ restricted as in \eqref{eq:FS:curve:constraint}.


\section{FDP control and power comparison} \label{sec:power}

In real-world scientific applications, together with the discovery of variables that are truly associated with the response, another important question statisticians need to answer is, what is the expected fraction of false discoveries among all discoveries? This is called false discovery rate (FDR) and the sample version is called false discovery proportion (FDP); many works have been developed to control these quantities. Although our method is not designed to control FDP, we produce selective $p$-values for the survived variables, and thus with the help of Type I error or FDP control under multiple testing framework, we can empirically check FDR control. To be more specific, the randomized Lasso with $\lambda_R^{cv}$ produces an active set $E$ with selective $p$-values. Now, we are conducting tests on $|E|$ hypotheses simultaneously. With methods controlling type I error or FDP, we further reduce the set of selected predictors $E$ to a set of predictors $E'$. Note that the set $E'$ denotes all the predictors that we selected via a model selection procedure and were rejected after performing a type I or FDR controlling method. To evaluate the performance, we will compare our results  on FDP with {\it knockoffs} \citep{original_knockoffs}, a popular method developed recently. Model-free knockoff of \citealt{candes2016panning} have been proposed recently, not requiring any distributional assumptions on $y|X$ and works also in $p>n$ regime unlike the original knockoff. However, the model-free knockoff method requires the knowledge of the distribution of the covariates which may be unrealistic in many applications. 
Comparison with \citealt{candes2016panning} are left for future work.
We note that our method is not designed with the same goal as knockoffs which is
explicitly designed to control FDR in such regression problems. Our method produce confidence intervals as well as variable specific $p$-values that knockoffs do not.
For comparison, we report FDP, Type I error and power. To clarify, 
\begin{equation*}
	\textnormal{FDP} = \frac{|\textnormal{false rejections}|}{|E'|}, \; \textnormal{Type I error} = \frac{|\textnormal{false rejections}|}{|E|}, \; \textnormal{Power} = \frac{|\textnormal{true rejections}|}{s},
\end{equation*}
where $s$ is the true sparsity.

In the current simulation, we illustrate two simple algorithms that attempt to control the overall Type I error or FDP (although we are aware of many other existing ways). 

\begin{itemize}[leftmargin=*]

\item {\it thresholding} at 0.05: we reject any hypothesis (variables in $E$) having selective $p$-values below 0.05, and the variables survived through thresholding rule compose the new set $E'$. In this way, the total number of false rejections is controlled under $0.05\cdot|E|$. In practice, when $|E|$ is small, type I error will be small even without adjusting for multiple testing.

\item {\it Benjamini-Hochberg (BH) with target FDR = 0.2} \citep{benjamini1995controlling}: we use the vanilla BH algorithm. Note that, in this case FDR will not be strictly controlled due to dependency between the selective $p$-values. Nevertheless, it still performs well empirically.

\end{itemize}

We perform both Lasso and randomized Lasso algorithm to get active set $E$, and we see the power increases when we go from non-randomized to randomized selective $p$-values; the latter procedure is comparable to knockoffs.

\textbf{Data generating mechanism.} We generate design matrix $X\in \mathbb{R}^{n\times p}$ from AR(1) model with auto-correlation $\rho$, i.e.~the rows of $X$ are taken to be independent from $\mathcal{N}_p(0, \Theta)$ with $\Theta_{jk}=\rho^{|j-k|}$, $j,k=1,\ldots,p$, (in case $\rho=0$, $\Theta=I_p$). We take $\rho$ to have values $0, 0.2, 0.4$. The columns of $X$ are then normalized to have empirical variance 1.
The noise vector $\epsilon\in\mathbb{R}^n$ is from $\epsilon \sim \mathcal{N}_n(0,I_n)$, independent of $X$. Then we generate response $y$ following the model $y=X\beta+\epsilon$.  The coefficient vector $\beta\in\mathbb{R}^p$ has true sparsity $s=30$, with non-zero coefficients having magnitude equal to 3.5 and signs $\pm 1$ following Bernoulli$(1/2)$. 
See \cite{original_knockoffs} for the details on this choice. Dimension-wise, we look at two scenarios: $\{n=3000, p=1000\}$ (this is the same setting as in \citealt{original_knockoffs}) and  \{$n=2000, p=1000$\}. Knockoff procedure becomes very conservative for $p\in[n/2,n]$; hence in the above two scenarios, we keep $p\leq n/2$ to have fair comparison; however, our methods can be used in $p>n$ settings, while the original knockoffs do not apply. Possible extensions of the original knockoffs to high dimensional setting via data splitting are in \citealt{high_dim_knockoffs}.

\begin{remark}
Theoretically, since $\lambda$s are chosen in data dependent ways, we have to adjust for cross-validation as we proposed. Empirically, we observe that for this specific data generating mechanism, the p-values with and without adjustment are similar. In other words, the p-values adjusting for model selection alone already look uniform. Therefore, to save computational cost, we stay with the non-adjusted version in this subsection.
\end{remark}

Along with knockoff, we compare our procedure with several other procedures that also give us $p$-values and confidence intervals as follows.

\begin{itemize}[leftmargin=*]

\item Data splitting (DS1): We use half of the data to get model $E$ through Lasso, with the penalty level chosen by cross-validation $\lambda=\lambda^{cv}$ as in \eqref{eq:cv:nonrandom} with no additional randomization.  Then, we compute $p$-values and confidence intervals for the least-square estimator constrained to set $E$ using the second half of the data.

\item Lee et al.~(TG1): We perform Lasso on the whole data set with $\lambda=\lambda_R^{cv}$ as in \eqref{eq:cv:random} chosen via randomized cross-validation. We take $Err_R=Err+R_1+R_2$, i.e.~the added randomization is additive, with $R_1\times R_2\sim\mathcal{N}_{|\Lambda|}(0, 0.01)\times \mathcal{N}_{|\Lambda|}(0, 0.01)$. CV curve $Err$ is computed using \verb|glmnet| \citep{friedmantibshirani}.
 Then the $p$-values are constructed by using truncated Gaussian of \citealt{lee2013exact} test statistic on the whole data set as in \eqref{eq:pivot} for the selected coefficients. 

\item Lee et al.~(TG2): This procedure is the same as (TG1) except that we perform Lasso on the whole data set with $\lambda=\lambda_R^{1\sigma}$ chosen via randomized CV followed by one sigma rule. One sigma rule in randomized cross-validation is defined as follows $\lambda_R^{1\sigma}=\{\max \lambda_l: \lambda_l\in\Lambda, Err_R(\lambda_l)\leq Err_R(\lambda_R^{cv})+SD_R(\lambda_R^{cv})\}$, where $SD_R(\lambda_R^{cv})$ corresponds to the standard error of the randomized CV curve $Err_R$ evaluated at $\lambda_R^{cv}$.

\item Randomized Lasso (R1): First we choose $\lambda=\lambda_R^{cv}$ by randomizing the cross-validation curve of non-randomized Lasso as in TG1, where $R_1\times R_2\sim\mathcal{N}_{|\Lambda|}(0, 0.01)\times \mathcal{N}_{|\Lambda|}(0, 0.01)$. Second, we perform randomized Lasso as in \eqref{eq:lasso:objective:randomized} on the whole data set with $\lambda=\lambda_R^{cv}$ to select the model $E$. Then we carry out MCMC sampling and inference based on the selective density as described in Section \ref{sec:randomized:lasso:cv} yielding $p$-values and confidence intervals for the selected coefficients in $E$.

\item Randomized Lasso (R2): This procedure is the same as (R2) except that $\lambda=\lambda_R^{1\sigma}$ is chosen using one sigma rule as in TG2.

\end{itemize}

Averaged over 100 repeated experiments, we summarize the results of running BH(0.2) on the selective $p$-values from (DS1, TG1, TG2, R1, R2) in Table \ref{table:BH:n:3000} and Table \ref{table:BH:n:2000} for $n=3000$ and $n=2000$, respectively.  We present empirical FDR, power and the average size of the selected set $|E|$ before BH. We compare our results with original knockoffs. Note that the size $|E|$ is not available (NA) for knockoffs since they only provide the final models.
We also present results from (DS1, TG1, TG2, R1, R2) in Table \ref{table:threholding:n:3000} and Table \ref{table:thresholding:n:2000}, for $n=3000$ and $n=2000$ respectively, where final models are determined by rule thresholding selective $p$-values at 0.05.


These four tables tell the following story. First, we gain greatly in terms of power when we use randomized Lasso instead of regular Lasso, while retaining valid inference. Second, in terms of power, R1 and R2 are comparable to knockoff for $n=3000$ and more powerful for $n=2000$. Third, with different choices of cross-validated $\lambda$, although the size of selected variables $|E|$ varies from case to case, power and FDR stay close.

It is worth noting that we used the parametric covariance estimates to make our procedures comparable to knockoffs since the latter is parametric as well. One more benefit from using the randomized inference approach (R1 and R2) is that, instead of relying heavily on parametric assumptions, as in both knockoffs and the truncated Gaussian statistic of \citealt{lee2013exact}, we are able to do non-parametric inference by using pairs bootstrap to estimate the covariances. Since we know the pairs bootstrap variance estimates are consistent pre-selection, using our post-selection consistency results we have that these variance estimates are also consistent post-selection.

Code used in this paper, including inference after cross-validated $\ell_1$-penalized logistic loss with different randomization distributions, is available online at \\ \verb|https://github.com/jonathan-taylor/selective-inference|.

\begin{table}[h!] 
\setlength{\tabcolsep}{6pt}
\centering
\begin{tabular}{c | c c c | c c c | c c c } 
\hline\hline
  & \multicolumn{3}{c}{\underline{$\rho=0$}} & \multicolumn{3}{c}{\underline{$\rho=0.2$}} &  \multicolumn{3}{c}{\underline{$\rho=0.4$}}  \\ [0.5ex] 
  & FDR & power & $|E|$ 
  & FDR & power & $|E|$
  & FDR & power & $|E|$ \\ [0.5ex] 
 \hline \hline
DS1 & 0.160 & 0.409 & 80.39
    & 0.159 & 0.420 & 78.88
    & 0.169 & 0.401 & 78.58 \\
TG1 & 0.068 & 0.135 & 124.17
    & 0.048 & 0.118 & 130.4
    & 0.068 & 0.135 & 124.17 \\
TG2 & 0.076 & 0.270 & 53.83
    & 0.071 & 0.331 & 48.24
    & 0.0705 & 0.283 & 47.11\\
R1 & 0.208 & 0.606 & 251.47
   & 0.184 & 0.601 & 254.51
   & 0.214 & 0.508 & 255.56\\
R2 & 0.196 & 0.573 & 120.26
   & 0.204 & 0.579 & 111.77
   & 0.245 & 0.538 & 121.95\\
knockoffs & 0.183 & 0.654 & 
 		  & 0.184 & 0.631 & 
 		  & 0.141 & 0.506 & \\
\hline
\end{tabular}%
\vspace{1em}
\caption{\textit{BH algorithm with target FDR = 0.2. n=3000.} }
\label{table:BH:n:3000}
\end{table}


\begin{table}[h!] 
\setlength{\tabcolsep}{6pt}
\centering
\begin{tabular}{c | c c c | c c c | c c c } 
\hline\hline
  & \multicolumn{3}{c}{\underline{$\rho=0$}} & \multicolumn{3}{c}{\underline{$\rho=0.2$}} &  \multicolumn{3}{c}{\underline{$\rho=0.4$}}  \\ [0.5ex] 
  & FDR & power & $|E|$
  & FDR & power & $|E|$
  & FDR & power & $|E|$ \\ [0.5ex] 
 \hline \hline
DS1 & 0.162 & 0.390 & 76.11
   & 0.148 & 0.379 & 83.11
   & 0.196 & 0.345 & 76.47 \\
TG1 & 0.068 & 0.104 & 117.19
    & 0.055 & 0.108 & 118.37
    & 0.053 & 0.119 & 108.95 \\
TG2 & 0.053 & 0.271 & 53.44
    & 0.071 & 0.253 & 54.63
    & 0.058 & 0.217 & 52.27 \\
R1 & 0.212 & 0.550 & 256.04
   & 0.200 & 0.541 & 245.69
   & 0.251 & 0.482 & 246.63\\
R2 & 0.231 & 0.579 & 135.49
   & 0.217 & 0.582 & 123.17
   & 0.273 & 0.544 & 132.72 \\
knockoffs & 0.130 & 0.503 &
 		  & 0.107 & 0.469 &
 		  & 0.093 & 0.346 & \\
\hline
\end{tabular}%
\vspace{1em}
\caption{\textit{BH algorithm with target FDR = 0.2. n=2000.} }
\label{table:BH:n:2000}
\end{table}


\begin{table}[h!] 
\setlength{\tabcolsep}{6pt}
\centering
\begin{tabular}{c | c c c | c c c| c c c } 
\hline\hline
  & \multicolumn{3}{c}{\underline{$\rho=0$}} & \multicolumn{3}{c}{\underline{$\rho=0.2$}} &  \multicolumn{3}{c}{\underline{$\rho=0.4$}}   \\ [0.5ex] 
   & FDR & power & Type I
   & FDR & power & Type I
   & FDR & power & Type I \\ [0.5ex] 
 \hline \hline
DS1  & 0.183 & 0.431 & 0.039
     & 0.170 & 0.437 & 0.036
     & 0.199 & 0.428 & 0.043 \\
TG1  & 0.347 & 0.225 & 0.031
     & 0.329 & 0.234 & 0.036
     & 0.347 & 0.205 & 0.032 \\
TG2  & 0.132 & 0.284 & 0.019
     & 0.108 & 0.317 & 0.024
     & 0.114 & 0.285 & 0.025 \\
R1  & 0.354 & 0.696 & 0.049
    & 0.346 & 0.699 & 0.047 
    & 0.390 & 0.627 & 0.050 \\
R2  & 0.225 & 0.591 & 0.049
    & 0.227 & 0.584 & 0.052
    & 0.277 & 0.574 & 0.061 \\
\hline
\end{tabular}%
\vspace{1em}
\caption{\textit{Thresholding $p$-values at 0.05. n=3000.} }
\label{table:threholding:n:3000}
\end{table}


\begin{table}[h!] 
\setlength{\tabcolsep}{6pt}
\centering
\begin{tabular}{c | c c c | c c c| c c c} 
\hline\hline
  & \multicolumn{3}{c}{\underline{$\rho=0$}} & \multicolumn{3}{c}{\underline{$\rho=0.2$}} &  \multicolumn{3}{c}{\underline{$\rho=0.4$}}   \\ [0.5ex] 
  & FDR & power & Type I
  & FDR & power & Type I
  & FDR & power  & Type I\\ [0.5ex] 
 \hline \hline
DS1 & 0.193 & 0.411 & 0.042 
    & 0.186 & 0.417 & 0.037 
    & 0.222 & 0.378 & 0.047 \\
TG1  & 0.404 & 0.188 & 0.036
     & 0.340 & 0.206 & 0.033
     & 0.354 & 0.203 & 0.037\\
TG2 & 0.119 & 0.296 & 0.019 
     & 0.117 & 0.289 & 0.018
    & 0.132 & 0.253  & 0.020 \\
R1 & 0.384 & 0.664 & 0.052 
    & 0.376 & 0.653 & 0.051
    & 0.406 & 0.596 & 0.052\\
R2  & 0.232 & 0.630 & 0.051
    & 0.196 & 0.730 & 0.0496
   & 0.158 & 0.670  & 0.064\\
\hline
\end{tabular}%
\vspace{1em}
\caption{\textit{Thresholding $p$-values at 0.05. n=2000.} }
\label{table:thresholding:n:2000}
\end{table}


\end{document}